\documentclass[graybox]{svmult}

\usepackage{mathptmx}
\usepackage{helvet}
\usepackage{courier}

\usepackage{makeidx}
\usepackage{graphicx}
\usepackage{multicol}
\usepackage[bottom]{footmisc}
\usepackage{amsmath,amssymb}

\newcommand{\g}{\mathfrak{g}}
\newcommand{\p}{\partial}
\newcommand{\bc}{\bar{c}}
\newcommand{\RR}{\mathbb{R}}
\newcommand{\CC}{\mathbb{C}}
\newcommand{\ZZ}{\mathbb{Z}}
\newcommand{\LL}{{\mathcal L}}
\newcommand{\M}{{\mathcal M}}

\newcommand{\ep}{\epsilon}
\newcommand{\pa}{\partial}
\newcommand{\Ga}{\Gamma}
\newcommand\mat[1]{\bigl(\begin{smallmatrix} #1 \end{smallmatrix}\bigr)}
\newcommand{\Pf}{\mbox{Pf}}
\newcommand{\cA}{\mathcal A}
\newcommand{\im}{\mbox{Im}}
\newcommand{\tr}{\operatorname{tr}}
\newcommand{\Ker}{\operatorname{Ker}}
\newcommand{\Tr}{\operatorname{Tr}}
\newcommand{\End}{\operatorname{End}}
\newcommand{\Aut}{\operatorname{Aut}}
\newcommand{\Div}{\operatorname{div}}
\newcommand{\sign}{\operatorname{sign}}

\makeindex             

\begin{document}

\title*{Lectures on quantization of gauge systems}

\author{Nicolai Reshetikhin}

\institute{Department of Mathematics, University
of California, Berkeley, CA 94720-3840, U.S.A. \at
KdV Institute for Mathematics, University of Amsterdam,
             Science Park 904
1098 XH Amsterdam, The Netherlands\at Niels Bohr Visiting Professor,
Department of Mathematics, University of Aarhus, Denmark\at
\email{reshetik@math.berkeley.edu}, \today}

\maketitle

\abstract{A gauge system is a classical field theory where among the
fields there are connections in a principal $G$-bundle over the
space-time manifold and the classical action is either invariant or transforms appropriately with respect to the action of the gauge
group.  The lectures are focused on the path integral
quantization of such systems.  Here two main examples of gauge
systems are Yang-Mills and Chern-Simons.}

\section{Introduction}

Gauge field theories are examples of classical field theories
with a degenerate action functional. The degeneration is due to
the action of an infinite-dimensional gauge group. Among most known examples are the Einstein gravity and Yang-Mills theory. The Faddeev-Popov (FP) method gives a recipe how to construct a quantization of a classical gauge field theory in terms of Feynman diagrams. Such quantizations are known as perturbative or semiclassical quantizations. The appearance of so-called ghost fermion fields is one of the important aspects of the FP method
\cite{FP}.

The ghost fermions appear in the FP approach as a certain technical
tool. Their natural algebraic meaning is clarified in the BRST
approach (the letters stand for Becchi, Rouet, Stora, and,
independently, Tyutin who discovered this formalism). In the BRST
setting, fields and ghost fermions are considered together as
coordinates on a super-manifold. Functions on this super-manifold
are interpreted as elements of the Chevalley complex of the Lie
algebra of gauge transformations. In this setting the FP action is a
specific cocycle and the fact that the integral with the FP action
is equal to the original integral with the degenerate action is a
version of the Lefschetz fixed point formula.

Among all gauge systems the Yang-Mills theory is most interesting
for physics because of its role in the standard model in high energy physics \cite{Wein}. At the moment there is a mathematically acceptable semiclassical (perturbative) definition of the Yang-Mills theory where the partition functions (amplitudes) are defined
as formal power series of Feynman integrals. The ultraviolet
divergencies in Feynman diagrams involving FP ghost fields
can be removed by renormalization \cite{tH}, and the corresponding
renormalization is asymptotically free \cite{Gr}. All these
properties make the Yang-Mills theory so important for
high energy physics.

A mathematically acceptable definition of the path integral in the $4$-dimensional Yang-Mills theory which goes beyond  perturbation theory is still an open problem. One possible
direction which may give such a definition is constructive
field theory, where the path integral is treated as a limit
of finite-dimensional approximations.

Nevertheless, even mathematically loosely defined, the path integral remains a powerful tool for {\it phenomenological} mathematical and physical research in quantum field theory. It predicted many interesting conjectures, many of which were proven later by rigorous methods.

The main goal of these notes is a survey of semiclassical quantization of the Yang-Mills and of the Chern-Simons theories.
These lectures can be considered as a brief introduction to the framework of quantum field theory (along the lines outlined
by Atiyah and Segal for topological and conformal field theories).
The emphases are given to the semiclassical quantization of classical field theories.

In the Einstein gravity the metric on space-time is a
field. It is well known in dimension four
that the semiclassical (perturbative) quantization of Einstein gravity fails to produces
renormalizable quantum field theory. It is also known that three-dimensional quantum gravity is related to the Chern-Simons theory for the non-compact Lie group $SL_2$. In this lectures we will not go as far as to discuss this theory, but will focus on the
quantum Chern-Simons field theory for compact Lie groups.

We start with a sketch of classical field theory, with some examples
such as a non-linear sigma model, the Yang-Mills theory, and the
Chern-Simons theory. Then we outline the framework of quantum field
theory following Atiyah's and Segal's descriptions of basic
structures in topological and in conformal field theories. The emphasis
is given to the semiclassical quantization. Then Feynman diagrams
are introduced in the example of finite-dimensional oscillatory
integrals. The Faddeev-Popov and BRST methods are first introduced
in the finite-dimensional setting.

The last two sections contain the definition of the semiclassical
quantization of the Yang-Mills and of the Chern-Simons theories. The
partition functions in such theories are given by formal power
series, where the coefficients are determined by Feynman diagrams.

These notes are based on lectures given when the author visited
Aarhus University in 2006-2007 and on lectures given during the
summer school at Roskilde University's S{\o}minestation Holb{\ae}k,
Denmark, May 2008. The notes benefited from the author's discussions
with many people. Special thanks to J. Andersen, A. Cattaneo, L. Faddeev, V.
Fock, D. Freed, T. Johnson-Freyd, D. Kazhdan, A. Losev, J. Lott, P. Mnev, A. Schwarz, and L. Takhtajan. The author is also grateful to the referees and to B. Boo{\ss}-Bavnbek for many corrections and suggestions which improved the original draft. Finally, the author thanks A. W. Budtz for helping with the figures.

This work was supported by the Danish National
Research Foundation through the
Niels Bohr initiative. The author is
grateful to Aarhus University for the
hospitality. The work was also supported by the NSF grant DMS-0601912 and by DARPA.

\section{Local Lagrangian classical field theory}

\subsection{Space-time categories}
Here we will focus on Lagrangian quantization of Lagrangian classical field theories.

In most general terms {\it objects} of a {\it $d$-dimensional
space-time category} \index{space-time category} are
$(d-1)$-dimensional manifolds (space manifolds). In specific
examples of space-time categories space manifolds are equipped with
a structure (orientation, symplectic structure, Riemannian metric, etc.).

A {\it morphism} between two space manifolds $\Sigma_1$ and
$\Sigma_2$ is a  $d$-dimensional manifold $M$, possibly with a
structure (orientation, symplectic, Riemannian metric, etc.), together with the identification of $\Sigma_1\sqcup \overline{\Sigma_2}$ with the boundary of $M$. Here $\overline{\Sigma}$ is the manifold
$\Sigma$ with reversed orientation.

{\it Composition} of morphisms is the gluing along the common boundary. Here are examples of space-time categories.

\vspace{0.5cm}

{\bf The $d$-dimensional topological category}. \index{topological
category} Objects are smooth, compact, oriented $(d-1)$-dimensional
manifolds. A morphism between $\Sigma_1$ and $\Sigma_2$ is the
homeomorphism class of $d$-dimensional compact oriented manifolds
with $\p M=\Sigma_1\sqcup \overline{\Sigma_2}$ with respect to
homeomorphisms  connected to the identity and constant at the boundary. The orientation on $M$
should agree with the orientations of $\Sigma_i$ in a natural way.

The composition consists of gluing two morphisms along the common boundary and then taking the homeomorphism class of the result
with respect to homeomorphisms constant at the remaining boundary.

\vspace{0.5cm}

{\bf The $d$-dimensional Riemannian category}. \index{Riemannian
category} Objects are $(d-1)$ Riemannian manifolds.  Morphisms
between two oriented $(d-1)$-dimensional Riemannian manifolds $N_1$
and $N_2$ are isometry classes of oriented $d$-dimensional Riemannian manifolds $M$ (with respect to the isometries constant
at the boundary),
such that $\pa M =N_1\sqcup \overline{N_2}$. The orientation on all
three manifolds should naturally agree, and the metric on $M$ agrees
with the metric on $N_1$ and $N_2$ on a collar of the boundary. The
composition is the gluing of such Riemannian cobordisms. For 
details see \cite{T-S}.

This category is important for many reasons. One of them is that
it is the underlying structure for statistical quantum field theories \cite{ID}.

\vspace{0.5cm}

{\bf The $d$-dimensional metrized cell complexes}. Objects are $(d-1)$-dimensional oriented metrized cell complexes (edges have length,
2-cells have area, etc.). A morphism between two such complexes
$C_1$ and $C_2$ is a metrized complex $C$ together with two
embeddings of metrized cell complexes $i: C_1\hookrightarrow C$,
$j: \overline{C_2}\hookrightarrow C$ where $i$ is orientation reversing and $j$ is orientation preserving. The composition is the gluing of
such triples along the common $(d-1)$-dimensional subcomplex.

This category has a natural subcategory which consists of
metrized cell approximations of Riemannian manifolds.

It is the underlying category for all lattice models
in statistical mechanics.

\vspace{0.5cm}

{\bf The Pseudo-Riemannian category} The difference between this category and the Riemannian category is that morphisms
are pseudo-Riemannian with the signature $(d-1,1)$.
This is the most interesting category for physics.
When $d=4$ it represents the space-time structure of our universe.

\subsection{Local Lagrangian classical field theory}

The basic ingredients of a $d$-dimensional local Lagrangian
classical field theory \index{local Lagrangian classical field
theory} are:
\begin{itemize}

\item For each space-time we assign the space of fields. Fields can be sections of a fiber bundle on a space-time, connections on a fiber bundle over a space-time, etc.

\item The dynamics of the theory is determined by a local Lagrangian. It assigns to a field a top volume form on $M$
    which depends locally on the field. Without giving a general definition we will give illustrating examples of local actions.
Assume that fields are functions $\phi: M\to F$, and that $F$ is a Riemannian manifold. An example of an local Lagrangian for a field theory in a Riemannian
category with such fields is
\begin{equation}\label{nl-sigma}
\LL(\phi(x), d\phi(x))=({1 \over 2} (d\phi(x), d\phi(x))_F-V(\phi(x)))dx,
\end{equation}
where $(.,.)_F$ is the metric on $F$, the scalar product on forms is induced by the metric on $M$, and $dx$ denotes
the Riemannian volume form on $M$.

 The action functional is the integral
    \[
    S_M[\phi]= \int_M \LL (\phi, d\phi).
    \]
Solutions to the Euler-Lagrange equations for $S_M$
form a (typically infinite-dimensional) manifold
$X_M$.

 \item A boundary condition is a constraint on boundary values
of fields which in ''good cases`` intersects with $X_M$ over a discrete set. In other words, there is a discrete set of
solutions to the Euler-Lagrange equations with given boundary conditions.

\end{itemize}

A $d$-dimensional classical field theory can be regarded as a functor
from the space-time category to the category of sets. It
assigns to a $(d-1)$-dimensional space the set of possible boundary values of fields, and
to a space-time the set of possible solutions to the
Euler-Lagrange equations with
these boundary values.

Some examples of local classical field theories are outlined
in the next sections.

\subsection{Classical mechanics}\label{cMech}
In classical mechanics space-time is a Riemannian 1-dimensional
manifold with flat metric, that is an interval. Fields in classical
Lagrangian mechanics \index{classical Lagrangian mechanics} are
smooth mappings of an interval of the real line to a smooth
finite-dimensional manifold $N$, called the {\it configuration
space} (parametrized paths).

The action in classical mechanics is determined by a choice of the Lagrangian function $\LL: TN\to \RR$ and is
\[
S_{[t_2,t_1]}[\gamma]=\int_0^t \LL(\dot \gamma(\tau), \gamma(\tau))d\tau,
\]
where $\gamma=\{\gamma(t)\}_{t_1}^{t_2}$ is a parametrized path in
$N$.

The Euler-Lagrange equations in terms of local coordinates $q=(q^1,\dots, q^n)\in N$ and $\xi=(\xi^1,\dots, \xi^n)\in T_qN$
are
\[
-\sum_{i=1}^n\frac{d}{dt} \frac{\pa \LL}{\pa \xi^i}(\dot\gamma(t), \gamma(t))+\frac{\pa \LL}{\pa q^i}(\dot\gamma(t), \gamma(t))=0
\]
where $\LL(\xi, q)$ is the value of the Lagrangian at the point $(\xi, q)\in TN$.

The Euler-Lagrange equations are a non-degenerate system of second order differential equations, if $\frac{\pa^2 \LL}{\pa \xi^i\pa \xi^j}(\xi,q)$ is non-degenerate for all $(\xi, q)$. In realistic
systems it is assumed to be positive.

Even when the Euler-Lagrange equations are satisfied, the variation of the action is still not necessarily vanishing. It is given by boundary terms:
\begin{equation}\label{b-var}
\delta S_{[t_2,t_1]}[\gamma]=\frac{\pa \LL}{\pa \xi^i}(\dot\gamma(t), \gamma(t))\delta\gamma(t)^i|_{t_1}^{t_2}\, .
\end{equation}

Imposing Dirichlet boundary conditions means fixing boundary points of the path: $\gamma(t_1)=q_1\in N$, and $\gamma(t_2)=q_2\in N$.
With these conditions the variation of $\gamma$ at
the boundary of the interval is zero and the boundary terms
in the variation of the action vanish.

A concrete example of  classical Lagrangian mechanics is the motion
of a point particle on a Riemannian manifold in the potential
force field. In this case
\begin{equation}\label{Lpp}
\LL(\xi, q)=\frac{m}{2}(\xi,\xi)+V(q),
\end{equation}
where $(.,.)$ is the metric on $N$ and $V(q)$ is the potential.

\subsection{First order classical mechanics}\label{1-ord-Lagmech}
The non-degeneracy condition \index{non-degeneracy condition} of
$\frac{\pa^2 \LL}{\pa \xi^i\pa \xi^j}(\xi,q)$ is violated in an
important class of first order Lagrangians.

Let $\alpha$ be a $1$-form on $N$ and $b$ be a function on $N$.
Define the action
\[
S_{[t_2,t_1]}[\gamma]=\int_{t_1}^{t_2} (\left< \alpha(\gamma(t)),\dot\gamma(t)\right> + b(\gamma(t)))dt,
\]
where $\gamma$ is a parametrized path.

The Euler-Lagrange equations for this action are:
\[
\omega(\dot\gamma(t))+db(\gamma(t))=0,
\]
where $\omega=d\alpha$. Naturally, the first order Lagrangian system
\index{first order Lagrangian system} is called {\it
non-degenerate}, if the form $\omega$ is non-degenerate. It is clear
that a non-degenerate first order Lagrangian system defines a
symplectic structure on a manifold $N$. The Euler-Lagrange equations
for such system are equations for flow lines of the Hamiltonian on
the symplectic manifold $(N, \omega)$ generated by the Hamiltonian
$H=-b$. 

Assuming that $\gamma$ satisfies the Euler-Lagrange equations
the variation of the action does not yet vanish. It is
given by the boundary terms (\ref{b-var}):
\begin{equation*}
\delta  S_{[t_2,t_1]}[\gamma]=\left< \alpha(\gamma(t)),\delta\gamma(t)\right> |_{t_1}^{t_2}\, .
\end{equation*}
If $\gamma(t_1)$ and $\gamma(t_2)$ are constrained to Lagrangian
submanifolds in $L_{1,2}\subset N$ with
$TL_{1,2}\subset\ker(\alpha)$, these terms vanish.

Thus, constraining boundary points of $\gamma$ to such a Lagrangian
submanifold is a natural boundary condition for non-degenerate first
order Lagrangian systems. As we will see, this is a part of the more
general concept where Lagrangian submanifolds define natural
boundary conditions for Hamiltonian systems.

\subsection{Scalar fields}\label{cBose} The space-time in such theory is a Riemannian category. Fields are smooth mappings from a space-time to $\RR$ (sections of the trivial fiber bundle $M\times \RR$). The action functional is
\[
S_M[\phi]=\int_M \left({1 \over 2} (d\phi(x), d\phi(x))-V(\phi(x))\right)dx,
\]
where the first term is determined by the metric on $M$
and $dx$ is the Riemannian volume form. The Euler-Lagrange equations are:
\begin{equation}\label{EUB}
\Delta \phi +V'(\phi)=0.
\end{equation}

The Dirichlet boundary conditions \index{Dirichlet boundary
conditions} fix the value of the field at the boundary $\phi|_{\pa
M}=\eta$ for some $\eta: \pa M\to \RR$. The normal derivative of the
field at the boundary varies for these boundary conditions.

\subsection{Pure Euclidean $d$-dimensional Yang-Mills}\label{cYM}
\index{Yang-Mills}
\subsubsection{Fields, the classical action, and the gauge invariance}
The space-time is a Riemannian $d$-dimensional manifold. Fields are connections on a principle $G$-bundle $P$ over $M$, where
$G$ is a compact Lie group (see for example \cite{Fr}) for basic definitions). Usually it is a simple (or Abelian) Lie group.

The action functional is given by the integral
\[
S_M[A]=\int_M \frac{1}{2} \tr \left< F(A),F(A)\right> dx,
\]
where $\left< .,.\right> $ is the scalar product of two-forms on $M$ induced by the metric, $\tr (AB)$ is the Killing form on the Lie algebra $\g=Lie(G)$, $F(A)$ is the curvature of $A$, and $dx$ is the volume form.

The Euler-Lagrange equations for the Yang-Mills action
are:
\[
d_A^* F(A)=0.
\]

The Yang-Mills action is invariant with respect to gauge
transformations. \index{gauge transformations} Recall that gauge
transformations are bundle automorphisms (see for example
\cite{Fr}). Locally, a gauge transformation acts on a connection as
\[
A\mapsto A^g=g^{-1}Ag+g^{-1}dg.
\]
Here we assume that $G$ is a matrix group and $g^{-1}dg$ is the
Maurer-Cartan form on $G$. Now let us describe the Dirichlet
boundary conditions \index{Dirichlet boundary conditions} for the
Yang-Mills theory. Fix a connection $A^b$ on $P|_{\pa M}$. The
Dirichlet boundary conditions on the connection $A$ for the
Yang-Mills theory require that $A^b$ is the pull-back of $A$ to the
boundary induced by the embedding $i: \pa M\to M$, i.e.
$i^*(A)=A^b$. Gauge classes of Dirichlet boundary conditions define
gauge classes of solutions to the Yang-Mills equations.
See \cite{FU} for more details about classical Yang-Mills theory.

\subsection{Yang-Mills field theory with matter} Let $V$ be a finite-dimensional representation of the Lie group $G$, and $V_P=P\times_G V$ be the
vector bundle over $M$ associated to  a principal $G$-bundle $P$. Assume that $V$ has
an invariant scalar product $(.,.)$.

The classical Yang-Mills theory with matter fields, which are sections of $V_P$, has the action functional
\[
S[\Phi, A]=\int_M \left(\frac{1}{2} \tr \left< F(A),F(A)\right> + \frac{1}{2}(\left< d_A\Phi, d_A\Phi\right> )+U(\Phi)\right)dx,
\]
where $U$ is a $G$-invariant function on $V$ and $\left< .,.\right> $ is the
scalar product on forms defined by the metric on $M$. The function $U$ describes the self-interaction of the scalar field $\Phi$.

The Euler-Lagrange equations in this theory are
\[
*d_A F(A) + j_A=0, \ \ \  d_A^*d_A\Phi-U'(\Phi)=0,
\]
where $j_A\in \Omega^1(M,\g)$ is the one-form defined as
$\tr \left< \omega, j_A\right> =\left< \omega \Phi, d_A \Phi\right> $.

Dirichlet boundary conditions in this theory are determined by the
gauge class of the boundary values of the connection $A$ and of the scalar field $\Phi$.

\subsection{$3$-dimensional Chern-Simons theory}\label{cCS}

In this case the space-time category is the category of
3-dimensional topological cobordisms. Fix a smooth 3-dimensional
manifold $M$. The space of fields of the Chern-Simons theory is the
space of connections on a trivial principal $G$-bundle $P$ over $M$
(just as in the Yang-Mills theory). The choice of a simple compact
Lie group $G$ is part of the data.

The Chern-Simons form is the 3-form on $P$:
\[
\alpha(A)=\tr \left(A\wedge dA-\frac{2}{3} A\wedge [A\wedge A]\right).
\]
Because the bundle is trivial, $\alpha(A)$ defines a 3-form on $M$
which we will also denote by $\alpha(A)$. The Chern-Simons action
\index{Chern-Simons action} is
\[
CS_M(A)=\int_M \alpha(A).
\]

This action is of the first order (in derivatives of $A$). It is
very different from the Yang-Mills theory
where the action is of the second order.

The variation of the Chern-Simons action is
\[
\delta CS_M(A)=\int_M \tr \left( F(A)\wedge \delta A\right)
+\int_{\pa M} \tr \left(A_\tau\wedge \delta A_\tau\right),
\]
where $A_\tau, \delta A_\tau$ are pull-backs to the boundary of
$A$ and $\delta A$.

The Euler-Lagrange equations for this Lagrangian are
\[
F(A)=0.
\]
They guarantee that the first term (the bulk) in the variation vanishes. Solutions to the Euler-Lagrange equations are flat connections in $P$ over $M$. On the space of solutions to the Euler-Lagrange equations we have
\begin{equation}\label{theta}
\delta CS_M(A)=(\Theta, \delta A_\tau),
\end{equation}
where $\Theta$ is a one form on the space $C_{\pa M}$ of connections
on $P|_{\pa M}\to \pa M$ defined by (\ref{theta}). Let $D$ be the differential
acting on forms on the space $C_{\pa M}$.
The form $\omega=D\Theta$ is non-degenerate and defines
a symplectic structure on $C_{\pa M}$ \cite{AB}:
\begin{equation}\label{s-conn}
\omega(\delta A, \delta B)=\int_{\pa M} \tr \delta A\wedge \delta B.
\end{equation}

The Chern-Simons action is invariant with respect to the action 
of the Lie algebra of gauge transformations ( for details 
on gauge invariane with respect to global gauge transformations see \cite{Fr}). The action of the gauge group is Hamiltonian on
$(C_{\pa M}, \omega)$. The result of the Hamiltonian reduction
of this symplectic space with respect to the action of
the gauge group is
the finite-dimensional moduli space $F(\pa M)$
of gauge classes of flat connections together with reduced
symplectic structure.

Gauge orbits through flat connections from $C_{\pa M}$ which
continue to flat connections on $P$ over $M$ form a Lagrangian
submanifold  $L_M\subset F(\pa M)$. The corresponding first order
Hamiltonian system describes the reduced Chern-Simons theory as a
classical Hamiltonian field theory. 

\section{Hamiltonian local classical field theory}
\index{Hamiltonian local classical field theory}
\subsection{The framework} An $n$-dimensional Hamiltonian field theory in a category of
space-time is an assignment of the following data to manifolds
which are the objects and morphisms of this category:
\begin{itemize}

\item A symplectic manifold $S(M_{n-1})$ to an $(n-1)$-dimensional
manifold $M_{n-1}$.

\item A Lagrangian submanifold $L(M_n)\subset S(\partial M_n)$ to
each $n$-dimensional manifold $M_n$.

\end{itemize}

These data shall satisfy the following axioms:
\begin{enumerate}

\item $S(\O)=\{0\}$.

\item $S(M_1\sqcup  M_2)= S(M_1)\times S(M_2)$.

\item $L(M_1\sqcup M_2)=L(M_1)\times L(M_2)$ with $L(M_i)\subset
S(\partial M_i)$.

\item $(S(\overline{M}),\omega)=(S(M),-\omega)$.

\item An orientation preserving diffeomorphism $f:M_1\to M_2$  of
$(n-1)$-dimensional manifolds lifts to a symplectomorphism $s(f):
S(M_1)\to S(M_2)$.

\item Assume that $\partial M=(\partial M)_1\sqcup (\partial
M)_2\sqcup (\partial M)'$ and that there is an orientation reversing
diffeomorphism $f:(\partial M)_1\to \overline{(\partial M)_2}$.
Denote by $M_f$ the result of gluing $M$ along $(\partial M)_1\simeq
\overline{(\partial M)_2}$ via $f$:
\[
M_f=M/\left< (\partial M)_1\simeq \overline{(\partial M)_2}\right> .
\]
The Lagrangian submanifold corresponding to the result of the gluing should be
\begin{multline}
L(M_f)=\{x\in S((\partial M)')|\mbox{ such that there exists }
y\in S(\partial M)_1 \\ \mbox{ with } (y, s(f)(y), x)\in L(M)\}.
\end{multline}
Notice that $\partial M_f=(\partial M)'$ by definition. This axiom
is known as the gluing axiom. In classical mechanics the gluing
axiom is the composition of the evolution at consecutive intervals
of time.\footnote{I am grateful to V. Fock for many illuminating
discussions of Hamiltonian aspects of field theory, see also
\cite{Fo}.}
\end{enumerate}

Note, that $M$ does not have to be connected.

A {\it boundary condition} in the Hamiltonian formulation is a
Lagrangian submanifold $L^b(\pa M)$ in the symplectic manifold
$S(\pa M)$, assigned to the boundary $\pa M$ of the manifold $M$,
$L^b(\pa M)\subset S(\pa M)$. It factorizes into the product of
Lagrangian submanifolds corresponding to connected components of the
boundary:
\[
L^b((\pa M)_1\sqcup (\pa M)_2 )=L^b((\pa M)_1)\times L^b((\pa M)_2).
\]
Classical solutions with given boundary conditions are intersection points $L^b(\pa M)\cap L(M)$.

In order to glue classical solutions along the common boundary
smoothly (composition of classical trajectories in classical mechanics) we assume that boundary Lagrangian submanifolds are fibers of Lagrangian fiber bundles. That is, we assume that for each connected
component $(\pa M)_i$ of the boundary a symplectic manifold $S((\pa
M)_i)$ is given together with a Lagrangian fiber bundle $\pi_i:
S((\pa M)_i)\to B((\pa M)_i)$ over some base space $B((\pa M)_i)$
with fibers defining the boundary conditions.

\subsection{Hamiltonian formulation of local Lagrangian
field theory}

Here again, instead of giving general definitions we will give a few illustrating examples.

\subsubsection{Classical Hamiltonian mechanics}\label{hMec}
\index{classical Hamiltonian mechanics}

{\bf 1.} Let $H\in C^\infty(M)$ be the Hamiltonian function generating
Hamiltonian dynamics on a symplectic manifold $M$\footnote{
Recall that Hamiltonian mechanics is a dynamical system on
a symplectic manifold $(M,\omega)$ with trajectories being flow lines of the Hamiltonian vector field $v_H$ generated by a function $H\in C^\infty(M)$, $v_H=\omega^{-1}(dH)$. Here $\omega^{-1}:T^*M\to TM$ is the isomorphism induced by the
symplectic structure on $M$.}. Here is
how such a system can be reformulated in the framework of a Hamiltonian field theory.

Objects of the corresponding space-time category are points;
morphisms are intervals $I=[t_1,t_2]\subset \RR$ with the flat
metric. The symplectic manifold assigned to the boundary of the
space-time is
\[
S(t_1,t_2)=\overline{M}\times M,
\]
where $M$ is the phase space of the Hamiltonian system and
$\overline{M}$ is the phase space with the opposite sign of the symplectic form.

The Lagrangian subspace $L(I)$ in $S(t_1,t_2)$ is the set of pairs
of points $(x,y)$ where $x$ is the initial point of a classical
trajectory generated by $H$ and $y$ is the target point of this
trajectory.

A pair of Lagrangian fiber bundles $\pi_1: M\to B_1$ and $\pi_2:
M\to B_2$ with suitable base spaces $B_1, B_2$ defines a ''complete``
family of boundary conditions corresponding to the two components
of the boundary of $I$.

Classical trajectories with such boundary conditions are
intersection points in $\bigl(\pi_1^{-1}(b_1)\times
\pi_2^{-1}(b_2)\bigr)\cap L(I)$, where $b_1\in B_1, b_2\in B_2$.

{\bf 2.} The Lagrangian mechanics on $N$ (see Section \ref{cMech})
is equivalent (for non-degenerate Lagrangians) \index{non-degenerate
Lagrangians} to the Hamiltonian mechanics on $M=T^*N$ with the
canonical symplectic form. The Hamiltonian functions are given by
the Legendre transform \index{Legendre transform} of the Lagrangian:
\[
H(p,q)=\max_{\xi\in T_qN}(p(\xi)-L(\xi,q)).
\]

The boundary conditions $q(t_1)=q_1, q(t_2)=q_2$ correspond
to Lagrangian fiber bundles $T^*N\to N$ for each component of the boundary of the interval.

The Hamiltonian of a point particle on a Riemannian
manifold is:
\[
H(p,q)=\frac{m}{2} (p,p)+V(q),
\]
where $(p,p)$ is uniquely determined by the metric on $N$.

{\bf 3.} A non-degenerate first order Lagrangian defines a symplectic structure on the configuration space $M$ given by $\omega=d\alpha$. Solutions to the Euler-Lagrange equations in such
system are flow lines of the Hamiltonian vector field generated by the function  $b(q)$, see Section \ref{1-ord-Lagmech}. So, first order non-degenerate Lagrangian systems are simply Hamiltonian systems on exact symplectic manifolds (i.e. on symplectic manifolds where the form $\omega$ is exact).

\subsubsection{Bose field theory} In this case the symplectic manifold $S(N)$ assigned to a $(d-1)$-dimensional manifold $N$ is an infinite-dimensional linear symplectic manifold which is the cotangent bundle to the space of real-valued smooth functions on $N$.

Since $C^\infty(N)$ is a linear space its tangent space at any point (can be thought as the space of infinitesimal
variations of functions on $N$) can be naturally identified
with the $C^\infty(N)$ itself. For the cotangent bundle to 
$C^\infty(N)$ we have $T^*C^\infty(N)=C^\infty(N)\oplus \Omega^{top}(N)$ with the symplectic form
\[
\omega((\delta\eta_1,\delta f_1),(\delta\eta_2, \delta f_2))=\int_N(\delta\eta_1 \delta f_2-\delta\eta_2 \delta f_1),
\]
where  $(\delta\eta_i, \delta f_i)$ are tangent vectors to $T^*C^\infty(N)$.

The Lagrangian fibration corresponding to the Dirichlet boundary conditions is the standard projection $\pi: T^*C^\infty(N)\to C^\infty(N)$.

The Lagrangian submanifold $L(M)\subset S(\pa M)$ is the space of
pairs $(f, \eta)$ where $f$ is the boundary values of a solution
$\phi$ to the Euler-Lagrange equations and $\eta=f_n dx$, $f_n$ is the normal derivative of $\phi$ at the boundary and $dx$ is the Riemannian volume form. Solutions to the
Euler-Lagrange equations with given Dirichlet boundary condition
$\phi|_{\pa M}=f$ are intersection points of $L(M)$ with the
Lagrangian fiber $\pi^{-1}(f)$.

\subsubsection{Yang-Mills theory} Here we will discuss only the Yang-Mills theory where fields are connections in a trivial
principal $G$-bundle. The natural symplectic manifold $\tilde S(N)$ assigned to the $(d-1)$-dimensional manifold $N$ in such field theory is the cotangent bundle to the space $\cA(N)$ of connections in a trivial principal
$G$-bundle over $N$ with the natural symplectic structure.
It can be identified naturally with $\Omega^1(N)\oplus \Omega^{d-1}(N)$ with the symplectic form
\[
\omega((\delta\eta_1,\delta A_1), (\delta\eta_2,\delta A_2))=\int_N
(\tr(\delta\eta_1\wedge\delta A_2) - \tr( \delta\eta_2\wedge
\delta A_1 ).
\]
Here we assume that $G$ is a matrix group and $tr(ab)$ is the Killing form, $(\delta\eta_i,\delta A_i)$ are pairs of $1$ and $d-1)$ $\g$-valued forms on $N$ which are tangent vectors to
the cotangent bundle to the space of connections. This is the so-called non-reduced phase space.

In the Hamiltonian formulation of the Yang-Mills theory, the
symplectic manifold $S(\pa M)$ is the Hamiltonian reduction of
$T^*\cA(\pa M)$ with respect to the action of the gauge group. 
This manifolds is also the cotangent bundle to the space of
gauge classes of connections on $G\times \pa M$. The
Lagrangian submanifold $L(M)\subset S(\pa M)$ is the 
image (with respect to the Hamiltonian reduction 
of the subspace of pairs $(\eta, a)$ where $a$ is the pull-back to the boundary of a solution $A$ to the Yang-Mills equations, and $\eta$ is the (d-2)-form which is the pull-back to the boundary of
$*dA$, here $*$ is the Hodge operation corresponding to the metric
on $M$ .

\subsubsection{Chern-Simons} The main difference between the Yang-Mills theory and the Chern-Simons theory is
that the YM theory is a second order theory while the
CS is a first order theory. Solutions to the Euler-Lagrange equations are flat connections on
$M$, and their pull-backs to the boundary are  flat connections on the boundary $\pa M$.

In the Hamiltonian formulation of the Chern-Simons theory,
\index{Hamiltonian formulation of the Chern-Simons theory} the
symplectic manifold assigned to the boundary is the moduli space of
flat connections on $P_{\pa M}$ with the symplectic structure
described in \cite{AB}. The Lagrangian submanifold
$L(M)$ is the space of gauge classes of flat connections on $P_{\pa
M}$ which continue to flat connections on $P$.

\section{Quantum field theory framework}

\subsection{General framework of quantum field theory}

We will follow the framework of local quantum field theory
\index{framework of local quantum field theory} which was outlined
by  Atiyah and Segal for topological and conformal field theories.
In a nut-shell it is a functor from a category of cobordisms to the
category of vector spaces (or, more generally to some ''known``
category).

All known local quantum field theories can be formulated in
this way at some very basic level. It does not mean that this is a final destination of our understanding of quantum dynamics at the microscopical scale. But at the moment this general setting includes the standard model, which agrees with most of the experimental data in high energy physics. In this sense this is the accepted framework at the moment, just as at different points of history classical mechanics, classical electro-magnetism, and quantum mechanics were playing such a role \footnote{The string theory goes beyond such framework and beyond scales of present
experiments. It is a necessary step further, and it already produced a number of outstanding mathematical ideas and
results. One of the differences between the string theory
and the quantum field theory is that the concept of non-perturbative string theory is still  developing.}.

A {\it quantum field theory} in a given space-time category can be
defined as a functor from this category to the category of vector
spaces (or to another `standard', `known' category). It assigns a
vector space to the boundary and a vector in this vector space to
the manifold:
\[
N\mapsto H(N), \ \ M\mapsto Z_M\in H(\pa M).
\]
The vector space assigned to the boundary is the space of pure
states of the system on $M$. It may depend on the extra structure at
the boundary (it can be a vector bundle over the moduli space of
such structures). The vector $Z(M)$ is called the {\it partition
function} or {\it the amplitude}.

These data should satisfy natural axioms, such as
\begin{eqnarray}
H(\O)=\CC\,,\quad H(N_1\sqcup N_2)&=&H(N_1)\otimes H(N_2), \text{ and}\\
Z_{M_1\sqcup M_2}&=&Z_{M_1}\otimes Z_{M_2}\in H(\pa M_1)\otimes
H(\pa M_2).
\end{eqnarray}
An isomorphism (in the relevant space-time category)
$f:N_1\to N_2$ lifts to a linear isomorphism
\[
\sigma(f): H(N_1)\to H(N_2).
\]
The pairing
\[
\left< .,.\right> _N: H(\overline{N})\otimes H(N)\to \CC
\]
is defined for each $N$. This pairing should agree with
partition functions in the following sense. Let $\pa M=N\sqcup\overline{N}\sqcup{N'}$, then
\begin{equation}\label{gluing-prop}
(\langle.,.\rangle\otimes id)Z_{M}=Z_{M_N}\in H(N')
\end{equation}
where $M_N$ is the result of gluing of $M$ along $N$.
The operation is known as the gluing axiom. We outlined its
structure. The precise definition involves more details (see
\cite{At1}, \cite{Se}). 

The gluing axiom in particular implies that $Z$ can be regarded as a functor from the corresponding space time category to the category of vector spaces. If $\pa M=\overline{N_2}\sqcup N_1$,
the corresponding partition function is a linear mapping
$Z(M): H(N_2)\to H(N_1)$. If $\pa M'=\overline{N_3}\sqcup N_2$
the gluing axiom implies: 
\[
Z(M'\cup_{N_2} M)=Z(M) Z(M')\,.
\]

Originally this framework was formulated by Atiyah and Segal for
topological and conformal field theories, but it is natural  to
extend it to more general and more realistic quantum field theories,
including the standard model.

This framework is very natural in models of statistical mechanics on cell complexes with open boundary conditions, also known as lattice models.

The main physical concept behind this framework is the locality
of the interaction. Indeed, we can cut our space-time manifold in
small pieces and the resulting partition function $Z_M$ in such framework will be the composition of partition
functions of small pieces. Thus, the theory is determined by
its structure on `small' space-time manifolds, or at `short
distances'. This is the concept of {\it locality}.

\subsection{Constructions of quantum field theory}
\subsubsection{Quantum mechanics}  Quantum mechanics
\index{quantum mechanics} fits into the framework of quantum field
theory as a one-dimensional example. One-dimensional space-time
category is the same as in classical Lagrangian mechanics.

In quantum mechanics of a point particle on a Riemannian manifold $N$ the vector space assigned to a point is $L_2(N)$ with the usual scalar product. The quantized Hamiltonian is
the second order differential operator acting in $L_2(N)$
\[
\hat{H}=-\frac{mh^2}{2}\Delta +V(q),
\]
where $\Delta$ is the Laplace operator on $N$, $V(q)$
is the potential, and $h$ is the Planck constant.

The operator
\begin{equation}\label{q-evol}
U_{t_2-t_1}=\exp(\frac{i}{h}\hat{H}(t_2-t_1))
\end{equation}
is known as the {\it propagator}, or {\it evolution operator}
in quantum mechanics. It is a unitary operator
in $L_2(N)$ (assume $N$ is compact and $V(q)$ is sufficiently good). It can be written as an integral operator:
\begin{equation}\label{evol-kernel}
U_{t_2-t_1}(f)(q)=\int_NU_{t_2-t_1}(q,q')f(q')dq',
\end{equation}
where $dq'$ is the volume measure on $N$ induced by the metric.

The kernel $U_{t}(q,q')$ is a solution to the Schr{\"o}dinger
equation
\begin{equation}\label{Sch-eqn}
(ih \frac{\pa }{\pa t} -\frac{h^2}{2m}\Delta  +V(q))U_{t}(q,q')=0
\end{equation}
for $t>0$ with the initial condition
\[
\lim_{t\to +0} U_{t}(q,q')=\delta(q,q').
\]

Quantum mechanics of a point particle on a Riemannian manifold $N$
viewed as a $1$-dimensional quantum field theory assigns the vector
space $L_2(N)$ to a point, and the vector
$Z(I)(q_1,q_2)=U_{t_2-t_1}(q_2,q_1)\in H(\pa I)=\overline{L_2(N)\otimes L_2(N)}$ to
the interval $[t_2,t_1]$. Here $\overline{L_2(N)\otimes L_2(N)}$
is a certain completion of the tensor product which can be identified with a space of operators in $L_2(N)$, for details see
any mathematically minded textbook on quantum mechanics, for example \cite{Takht}.
For a variety of reasons (see \cite{BW}) it is better to think about the space
attached to a point not as $L_2(N)$ but as the space of 1/2-densities on $N$. Given two 1/2-densities $a$ and $b$, their
scalar product is
\[
(a,b)=\int_N \bar{a}b,
\]
where $\bar{a}b$ is now a density and can be integrated over $N$
(for details see for example \cite{BW}). In terms of 1/2-densities
the kernel of the evolution operator is a 1/2-density on $N\times N$ and
\[
U_t(a)(q)=\int_NU_t(q,q')a(q'),
\]
where $U_t(q,q')a(q')$ is a density in $q'$ and can be integrated over $N$.

\subsubsection{Statistical mechanics}
{\it Lattice models} in statistical mechanics also fit naturally
in the framework of quantum field theory. The space-time category
corresponding to these models is a combinatorial space category
of cell complexes.

A simple combinatorial example of combinatorial quantum
field theory with the dimer partition function can be found
in \cite{CimResh}.

The combinatorial construction of the TQFT (Topological Quantum
Field Theory) based on representation theory of quantized
universal enveloping algebras at roots of unity is given
in \cite{ReshTur} or,  more generally, on any modular
category.

Another combinatorial construction of TQFT, based on triangulations is given in \cite{TV}. This TQFT is the double of the construction
from \cite{ReshTur}, for details see for example \cite{Tu}.

\subsubsection{Path integral and the semiclassical quantization}
\index{path integral} If we were able to integrate over the space of
fields in a Lagrangian classical field theory (as in lattice models
in statistical mechanics) we could construct a quantization of a
$d$-dimensional classical Lagrangian system as follows:
\begin{itemize}
\item To a $(d-1)$-dimensional manifold we assign the space of functionals on boundary values of fields. Here we assume that a choice
    of boundary conditions was made.
\item To a $d$-dimensional manifold we assign a functional on    boundary fields given by the integral
    \[
    Z_M(b)=\int_{\phi|_{\pa M}=b} \exp(\frac{iS[\phi]}{h}) D\phi.
    \]
    If one treats the integral as a formal symbol which
    satisfies Fubini's theorem (the iterated integral is equal to the double integral),
    such assignment satisfies all properties of QFT. The problem is that the integral
    is usually not defined, unless the space of fields is finite or finite-dimensional
    (as in statistical mechanics of cell complexes).
    Thus, one should either make sense
    of the integral and check whether the definition satisfies Fubini's theorem, or define the QFT by some other means.
\end{itemize}

There are two approaches on how to make sense of {\it path integrals}.
The approach of {\it constructive field theory}, is based on
approximating the path integral by a finite-dimensional
integral and then proving that the finite QFT has a limit, when the mesh of the
approximation goes to zero.
For details of this approach see for example in \cite{GJ}.

Another approach is known as {\it perturbation theory}, or {\it semiclassical
limit}. The main idea is to define the path integral in the way its
asymptotic expansion as $h\to 0$ would look like, if the integral
were defined. The coefficients of this asymptotic expansion are
given by Feynman diagrams. Under the right assumptions the first few
coefficients would approximate the desired quantity sufficiently
well. The numbers derived from this approach are the base for the
comparison of quantum field theoretical models of particles with the
experiment.

In the next sections we will outline this approach on several examples.

When $M$ is a cylinder $M=[t_1,t_2]\times N$,
the partition function $Z_M$ is an element of $H(N)\times H(N)^*$
\footnote{ In this rather general discussion of the basic structures of a local quantum field theory we are deliberately somewhat vague about such details as the completion of the tensor
product and similar topological questions. Such questions are better answered on a case-by-case basis.} and therefore can be regarded as an operator in $H(N)$. Classical observables become
operators acting in $H(N)$. Thus, a quantization of
classical field theories for space-time cylinders can be regarded as passing from classical commutative observables to quantum non-commutative observables. The partition function for the
torus has a natural interpretation as a trace of the
partition function for the cylinder
(see for example \cite{GJ} for more details).

\section{Feynman diagrams}
\index{Feynman diagrams}

\subsection{Formal asymptotic of oscillatory integrals}
\index{formal asymptotic} \index{oscillatory integrals}
 Let $\M$ be a compact smooth manifold with a volume form on
it. In this section we will recall the diagrammatic formula for the
asymptotic expansion of the integral
\begin{equation}\label{int-f}
I_h(f)=\int_\M \exp\left(i\frac{f(x)}{h}\right) dx,
\end{equation}
where $f$ is a smooth function on $\M$ with finitely many
isolated critical points.

\begin{lemma}\label{Wick} We have the following identity
\begin{multline}
\lim_{\epsilon\to 0}\int_{\RR^N} \exp({i(x,Bx)/2-\epsilon (x,x)}) x_{i_1}\cdots x_{i_n} d^N x\\
 =
 (2\pi)^{\frac{N}{2}}i^{\frac{n}{2}}\frac{1}{\sqrt{|\det(B)|}}\exp({\frac{i\pi}{4}\sign(B)})
\sum_m B^{-1}_{i_{m_1}i_{m_2}}B^{-1}_{i_{m_1}i_{m_2}}\dots
B^{-1}_{i_{m_{n-1}i_{m_n}}}\, .
\end{multline}
Here the sum is taken over perfect matchings $m$ on the set
$\{1,2,\dots, n\}$, $\sign(B)$ denotes the {\it signature} of the
real symmetric matrix $B$ (the number of positive eigenvalues minus
the number of negative eigenvalues).

Moreover, if $n$ is odd, this integral is zero.
\end{lemma}
\begin{proof}
First notice that:
\[
\lim_{\epsilon\to 0}\int e^{\frac{i}{2}(x,Bx)-\epsilon (x,x)}
x_{i_1}\cdots x_{i_n} d^N x = \lim_{\epsilon\to 0}\frac{\pa}{\pa y_{i_1}}\cdots
\frac{\pa}{\pa y_{i_n}} \int e^{\frac{i}{2}(x,Bx)-\epsilon (x,x)+(y,x)}
d^Nx\Bigr|_{y=0}\, .
\]
After change of variables $z=x-iB^{-1}y$ in the Gaussian integral
\[
\lim_{\epsilon\to 0}\int_{\RR^N}\exp({\frac{i}{2}(x,Bx)}-\epsilon (x,x))
d^Nx=(2\pi)^{\frac{N}{2}}\frac{1}{\sqrt{|\det(B)|}}
\exp({\frac{i\pi}{4}\sign(B)})
\]
we have:
\begin{equation*}
\lim_{\epsilon\to 0}\int_{\RR^N}\exp({\frac{i}{2}(x,Bx)-\epsilon (x,x)+(y,x)})
d^Nx=(2\pi)^{\frac{N}{2}}\frac{1}{\sqrt{|det(B)|}}
\exp({\frac{i\pi}{4}\sign(B)}) \exp({\frac{i}{2}(B^{-1}y,y)})\, .
\end{equation*}
Expanding the right side in powers of $y$ we obtain the contribution of monomials of degree $2k$.
\begin{multline*}
\frac{i^k}{2^k!}\sum_{(i)(j)}(B^{-1})_{i_1j_1}\dots (B^{-1})_{i_{2k}j_{2k}}y_{i_1}\dots y_{i_k}y_{j_1}\dots y_{j_k}=
\frac{i^k}{2^k!}\sum_{i_1\leq \dots \leq
i_{2k}} \frac{y_{i_1}\dots y_{i_{2k}}}{m_1(i)!\dots m_{2k}(i)!}
\\ \sum_{\sigma\in S_{2k}}(B^{-1})_{\sigma(i_1)\sigma(i_2)}
(B^{-1})_{\sigma(i_3)\sigma(i_4)}\dots (B^{-1})_{\sigma(i_{2k-1})\sigma(i_{2k})}\, .
\end{multline*}
Here $m_1(i)$ is the number of the smallest entries in the
sequence $i_1,\dots, i_{2k}$, $m_2(i)$ is the number of the
smallest entries after the elimination of $i_1$ etc.

Taking derivatives with respect to $y$ and taking into account
that
\begin{multline*}
\frac{1}{2^k!}\sum_{\sigma\in
S_{2k}}(B^{-1})_{\sigma(i_1)\sigma(i_2)}
(B^{-1})_{\sigma(i_3)\sigma(i_4)}\dots (B^{-1})_{\sigma(i_{2k-1})\sigma(i_{2k})}=\\
\sum_m (B^{-1})_{i_{m_1}i_{m_2}}(B^{-1})_{i_{m_1}i_{m_2}}\dots
(B^{-1})_{i_{m_{2k-1}i_{m_{2k}}}}\, ,
\end{multline*}
where the sum is taken over perfect matchings on the set
$\{1,2,,\dots, 2k\}$, we obtain the desired formula.
\end{proof}

For example when $n=4$, then this integral is equal to
\[(B^{-1})_{12}(B^{-1})_{34}+(B^{-1})_{13}(B^{-1})_{24}+(B^{-1})_{14}(B^{-1})_{23}.\]
These three terms correspond to the perfect matching shown in Fig. \ref{4bosons}.

\begin{figure}[htb]
\includegraphics[height=2.5cm,width=8cm]{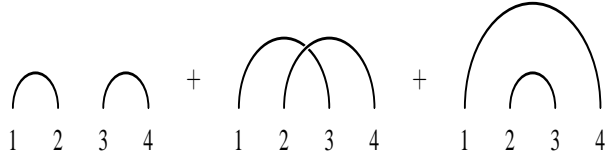}
\caption{Perfect matching for $n=4$}
\label{4bosons}
\end{figure}

\begin{theorem}
We have the following identity of power series
\begin{multline*}
\int_{\RR^N} \exp(i(x,Bx)/2-\sum_{n\ge 3}\frac{i}{n!} V^{(n)}(x)h^{n/2-1})\, d^Nx =\\
(2\pi)^{\frac{N}{2}}\frac{1}{\sqrt{|\det(B)|}}\exp({\frac{i\pi}{4}\sign(B)})
\sum_{\Ga}\frac{(ih)^{-\chi(\Ga)}F(\Ga)}{|\Aut (\Ga)|},
\end{multline*}
where the sum is taken over graphs with vertices of valency $\geq 3$, $F(\Ga)$ is the state sum corresponding to $\Ga$ described
below, $|\mbox{Aut} (\Ga)|$ is the number of elements in the
automorphism group of $\Ga$, $\chi(\Ga)=|V|-|E|$ is the Euler characteristic of the graph, $|E|$ is the number of edges of $\Ga$ and $|V|$ is the number of vertices of $\Ga$.
\end{theorem}
\begin{proof} Expanding the integral in formal power series in $h$ we have:
\begin{multline}\label{int-1}
 \int_{\RR^N} e^{\frac{i}{2}(x,Bx)+\sum_{n\ge 3}\frac{i}{n!} V^{(n)}(x)h^{n/2-1}}\, dx =
 \sum_{n_3\ge 0, n_4\ge 0\cdots} \frac{h^{(3n_3+4n_4+\cdots)/2 -n_3-n_4-\cdots}i^{n_3+n_4+\dots}}{n_3!(3!)^{n_3} n_4! (4!)^{n_4}\cdots}\\
 \int_{\RR^N} e^{i(x,Bx)/2} (V^{(3)}(x))^{n_3}(V^{(4)}(x))^{n_4}\cdots d^Nx.
\end{multline}
Here
\[
V^{(n)}(x)=\sum_{i_1,\dots,i_n} V^{(n)}_{i_1,\dots,i_n}x^{i_1}\dots x^{i_n}\, .
\]

For a graph $\Ga$ define the state sum $F(\Ga)$ as follows.
\begin{itemize}
\item Enumerate vertices, for each vertex enumerate edges adjacent to it. This defines a total ordering on
    endpoints of edges (the ordering from left to right
    in Fig. \ref{diag-match}).
\item The graph $\Ga$ defines a perfect matching between
edges adjacent to vertices as it is shown in Fig. \ref{diag-match}. Denote by $\Ga_m$ the graph corresponding to the perfect matching $m$.

\item Assign indices $i_1, i_2, \dots $ to endpoints
of edges, $i_\alpha=1,2,\dots, N$.

\item Define $F(\Ga)$ as
\[
F(\Ga)=\sum_{\{i\}}\prod_{e\in E(\Ga_m)}(B^{-1})_{e_l,e_r}V^{(n_1)}_{i_1,\dots, i_{n_1}}
V^{(n_1)}_{i_{n_1+1},\dots, i_{n_1+n_2}}V^{(n_1+n_2+1)}_{i_1,\dots, i_{n_1+n_2+n_3}}\dots
\]
\end{itemize}
where $e_l$ is the index corresponding to the left end of the edge
$e$, $e_r$ corresponds to the right side. The state sum $F(\Ga)$ is
the sum over $\{i\}$ of the product of weights assigned to vertices
and edges according to the rules from Fig. \ref{weights} and Fig.
\ref{ex-match}.\footnote{ Equivalently $F(\Ga)$ can be defined as
follows. Assign elements $1,\dots, N$ to endpoints of edges of
$\Ga$. This defines an assignment of indices to endpoints of stars
of vertices. The state sum is defined as
\[
F(\Ga)=\sum_{\{i\}} \prod_{e\in E(\Ga)} (B^{-1})_{i_e,j_e}
\prod_{v\in V(\Ga)}(\mbox{ weight of } v)_{i}\, .
\]
Here weights of vertices are defined as in Fig. \ref{weights},
the indices $i_e,j_e$ correspond to two different endpoints
of $e$ (since $B$ is symmetric, it does not matter that
this pair is defined up to a permutation).}

\begin{figure}[htb]
\sidecaption
\includegraphics[height=3cm,width=7cm]{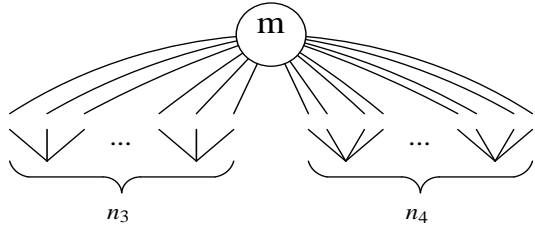}
\caption{Perfect matchings and Feynman diagrams}
\label{diag-match}
\end{figure}

\begin{figure}[htb]
\sidecaption
\includegraphics[height=3cm,width=7cm]{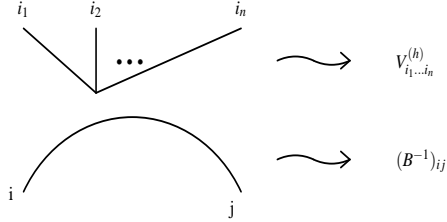}
\caption{Weights of vertices and edges in Feynman diagrams}
\label{weights}
\end{figure}

Lemma \ref{Wick} gives the following expression for
(\ref{int-1})
\begin{multline} \label{f-mat}
(2\pi)^{\frac{N}{2}}\frac{1}{\sqrt{|\det(B)|}}e^{\frac{i\pi}{4}\sign(B)}\,
i^{|V|} 
\sum_{n_3\ge 0, n_4\ge 0\cdots}
\frac{(ih)^{|E|-|V|}}{n_3!(3!)^{n_3} n_4! (4!)^{n_4}\cdots} \sum_m
F(\Gamma_m).
\end{multline}
Here the sum is taken over perfect matchings, and $\Gamma_m$ is the graph corresponding to the matching $m$, see Fig. \ref{diag-match}, $|E|$ is the number of edges and $|V|$ is the number of vertices of the graph $\Gamma_m$.

\begin{figure}[htb]
\includegraphics[height=2.5cm,width=8cm]{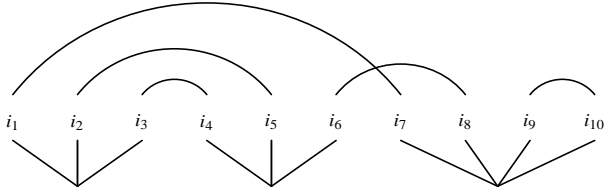}
\caption{An example of a perfect matching with a state $\{i\}$}
\label{ex-match}
\end{figure}

Some perfect matchings produce the same graphs. Denote by $N(\Ga)$ the number of perfect matchings corresponding to $\Ga$. In the formula (\ref{f-mat}) the contribution from the diagram $\Ga$
will have the combinatorial factor
\[
\frac{N(\Ga)}{n_3!(3!)^{n_3} n_4! (4!)^{n_4}\cdots}=\frac{1}{|\Aut(\Ga)|}.
\]
This finishes the proof.

There is a simple rule how to check powers of $i=\sqrt{-1}$.
These factors disappear, if we replace $B\mapsto iB$ and $V^{(n)}\mapsto iV^{(n)}$.
\end{proof}

A Feynman diagram {\it has order} $n$ if it appears as a coefficient
in $h^n$, i.e. when $n=|E|-|V|$ (or $n=-\chi(\Ga)$) in the expansion
above. As an example, order one Feynman diagrams are given in Fig.
\ref{ord-1} and Fig. \ref{ord-1-weights}.

\begin{figure}[htb]
\sidecaption
\includegraphics[height=1cm,width=7cm]{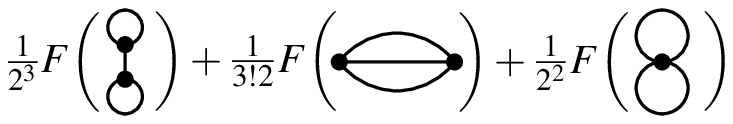}
\caption{Contributions from Feynman diagrams of order one}
\label{ord-1}
\end{figure}

\begin{figure}[htb]
\sidecaption
\includegraphics[height=2.5cm,width=8cm]{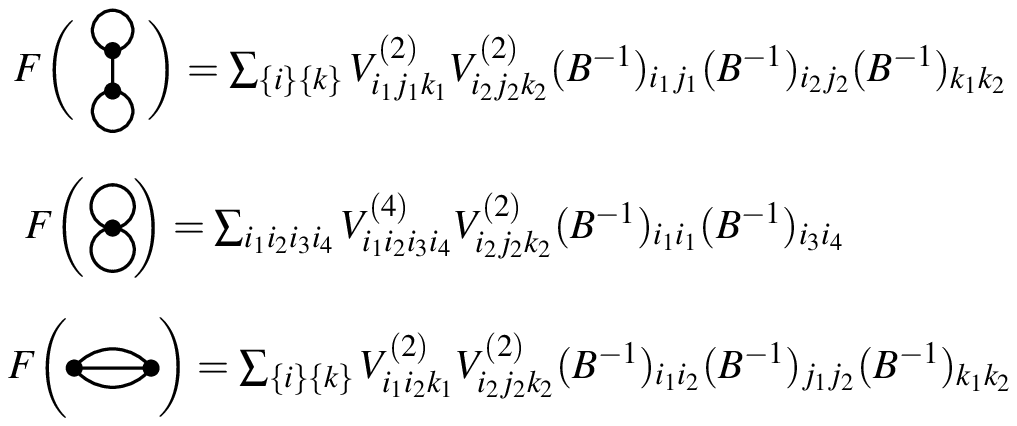}
\caption{Weights of Feynman diagrams of order one}
\label{ord-1-weights}
\end{figure}

Now let us focus on the asymptotic expansion of the integral
(\ref{int-f}). The standard asymptotic analysis applied to this
integral shows that the leading contributions to the asymptotics of the integral as $h\to 0$ come from the infinitesimal (of the diameter of
order $h^{-1/2}$) neighborhoods of critical points of $f(x)$. The
contribution to the integral (\ref{int-f}) from the critical point
$a$ ''localizes`` to the integral (\ref{int-1}) with
$(B_a)_{ij}=\frac{\pa^2f}{\pa x^i \pa x^j}(a)$ and
$(V_a^{(n)})_{i_1,\dots, i_n}=-\frac{\pa^n f}{\pa x^{i_1}\dots
x^{i_n}}(a)$.

Choose local coordinates such that $dx=dx_1\dots dx_N$. Denote by
$F_a(\Ga)$ the state sum on the graph $\Ga$ with such matrices $B$
and $V{(n)}$. The asymptotic expansion of the integral (\ref{int-f})
has the following form:
\begin{equation}\label{as-exp}
\int_\M \exp\left(i\frac{f(x)}{h}\right) dx\simeq \sum_a(2\pi
h)^{\frac{N}{2}}
\frac{1}{\sqrt{|\det(B_a)|}}e^{\frac{if(a)}{h}+\frac{i\pi}{4}\sign(B_a)}
\sum_{\Ga}\frac{(ih)^{-\chi(\Ga)}F_a(\Ga)}{|\Aut (\Ga)|}
\end{equation}
Here $\simeq$ is the asymptotical equivalence when $h\to 0$.
A similar argument applied to the integral
\begin{equation}\label{g-int}
\int_\M \exp(i\frac{f(x)}{h})g(x)dx
\end{equation}
gives the asymptotic expansion as $h\to 0$. It looks exactly as
(\ref{as-exp}) with the only difference that in each Feynman diagram there will be exactly one of the vertices given in Fig
\ref{g-vert}. The order of the diagram  is still $|E|-|V|$, where $V$ is the number of vertices given by derivatives of $f$, i.e. $-\chi(\Ga)$.

\begin{figure}[htb]
\sidecaption
\includegraphics[height=2cm,width=6cm]{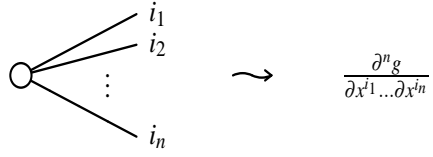}
\caption{Extra vertices in Feynman diagrams for the integral (\ref{g-int})}
\label{g-vert}
\end{figure}

\subsection{Integrals over Grassmann algebras}
\index{integrals over Grassmann algebras} \index{Grassmann algebra}
The {\it Grassmann algebra} $G_n$ is the exterior algebra of
$\CC^n$,
 $G_n=\wedge^{\cdot}\CC^n$ with the multiplication $(a,b) \to a\wedge b$. As an algebra defined in terms of generators and
 relations $G_n$ is generated by $c^1,\dots, c^n$ with defining relations $c^ic^j+c^jc^i=0$. The Grassmann algebra $G_n$ can also be regarded as the space of polynomial functions on the super-vector space $\CC^{0|n}$.

Left derivatives with respect to $c^i$ are defined as
\[
\pa{}_{c^i} c^{i_1}\cdots c^{i_n}=\begin{cases}
 0 & i\not\in \{i_1,\dots, i_n\}\\
 (-1)^kc^{i_1}\cdots \hat c^{i_k}\cdots c^{i_n} & i=i_k\, .
\end{cases}\]
The right derivatives are defined similarly with the sign $(-1)^{n-k}$ instead.

Recall that an {\it orientation} of $\CC^n$ is defined by a basis in $\bigwedge^n \CC^n$. Choose $c^1\wedge\cdots\wedge c^n$ as such orientation. Any element $P\in G_n$ can be written  as $p^{top}c^1\wedge\cdots\wedge c^n+$ lower terms. The integral of $P$ over the super-vector space $\CC^{0|n}$ with the orientation $c_1\wedge\dots \wedge c_n$ is
\[
 \int_{\CC^{0|n}} P\, dc:= p^{top}.
\]

\begin{lemma} Let $(c,Bc)=\sum_{ij=1}^n c^i B_{ij} c^j$, where $B$ is skew-symmetric $B_{ij}=-B_{ji}$. If $n$ is even, then
\begin{equation}\label{Pf}
  \int_{\CC^{0|n}} \exp\Bigl( \frac 12 (c,Bc)\Bigr)\, dc = \operatorname{Pf}(B),
\end{equation}
where $\operatorname{Pf}$ is the Pfaffian of the matrix $B$. If $n$ is odd,
the integral is zero.
\end{lemma}
\begin{proof} Recall that
\[
\operatorname{Pf}(B)=\sum_{m} (-1)^m B_{i_1j_1}B_{i_2j_2}\dots B_{i_{n/2}j_{n/2}}
\]
where the sum is taken over perfect matchings $m$. A perfect matching $m$ is the equivalence class
of a collection of pairs $((i_1,j_1),\dots, (i_{n/2},j_{n/2}))$ obtained by a permutation $\sigma$ of $(1,2,\dots, n)$ with respect to
permutations of pairs ($(i_a,j_a)$ with $(i_b,j_b)$) and permutations in a pair ($(i_a,j_a)$ to $(j_a,i_a)$). The sign
$(-1)^m$ is the sign of the permutation $\sigma$, which is
constant on the equivalence class $m$.

Now let us prove the formula (\ref{Pf}). It is clear that only
monomials of degree $n$ in $c$ will give a non-zero contribution to the integral and that they all come from the term
\[
(c,Bc)^{n/2} = \sum_{i_1,\dots, i_{n/2},j_1,\dots, j_{n/2}=1}^n B_{i_1j_1}\cdots B_{i_{n/2}j_{n/2}} c^{i^1}c^{j_1}\cdots c^{i_{n/2}}c^{j_{n/2}}.
\]
Reordering factors we get
\[
c^{i_1}c^{j_1}\cdots c^{i_{n/2}}c^{j_{n/2}} = (-1)^{\sigma(i|j)} c^1 \cdots  c^n\,,
\]
where $\sigma(i|j)$ is the permutation which brings $i_1,j_1,\dots,i_{n/2},j_{n/2}$ to $1,2,\dots, n$. Thus for the Gaussian Grassmann integral we have:
\[
\int_{\CC^{0|n}} \exp\Bigl( \frac 12 (c,Bc)\Bigr)\, dc =   \frac{(1/2)^{n/2}}{(n/2)!} \sum_{\sigma(i|j)} B_{i_1j_1}\cdots B_{i_{n/2}j_{n/2}} (-1)^{\sigma(i|j)}\, .
 \]
Note that the sign doesn't change when  $i_a$ is switched with $j_a$ because the signs come in pairs. Also, the sign doesn't change when  pair $(i_a,j_a)$ and $(i_b,j_b)$ are permuted.
But such equivalence classes of permutations are exactly perfect matchings and therefore the formula becomes
\begin{equation*}
  \sum_{\substack{\sigma(i|j)\\ i_a<j_a\\ i_{a_1}<\cdots < i_{a_n}}} (-1)^{\sigma(i|j)} B_{i_1j_1}\cdots B_{i_{n/2}j_{n/2}}
 = \operatorname{Pf}(B),
 \end{equation*}
 which is the Pfaffian of $B$.
 \end{proof}

This lemma is equivalent to the identity
 \[
  \Bigl(\sum_{i< j} x^i\wedge x^j B_{ij}\Bigr)^{\bigwedge \frac n2} = \operatorname{Pf}(B) x^1\wedge \cdots \wedge x^n
 \]
in the exterior algebra $\bigwedge^\cdot \CC^n$.

Two important identities for Pfaffians:
\[
\det B=\Pf(B)^2, \ \  \ \ \ \ \Pf\mat{0 & A\\ -A^t& 0} = \det A.
\]

The following formula is a Grassmann analog of the formula from
Lemma \ref{Wick} for integrating monomials with respect to the
Gaussian measure
\begin{equation}\label{Pf-corr}
  \int_{\CC^{0|n}} \exp\Bigl( \frac 12 (c,Bc)\Bigr)c^{i_1}\dots c^{i_k}\, dc = \Pf(B)(-1)^{\frac{k}{2}}\sum_m (-1)^m (B^{-1})^{i_{m_1}i_{m_2}}\dots (B^{-1})^{i_{m_{k-1}},i_{i_{m_k}}}.
\end{equation}
Here the sum is taken over perfect matchings $m$ of $1,\dots, k$, and $B$ is assumed to be non-degenerate. The proof of
this formula is parallel to the one for Gaussian oscillating integrals. The only difference is the factor $(-1)^m$ which
appears when left derivatives are applied to the exponent.

Let $P(c)$ be an even element of $G_n$ with monomials of degree at least $4$, $P(c)=\sum_{k\geq 4} \frac{1}{k!}P^{(k)}(c)$ where $P^{(k)}(c)=\sum_{i_1,\dots, i_k=1}^nP^{(k)}_{i_1,\dots, i_k}
c^{i_1}\dots c^{i_k}$.

\begin{theorem} The following identity holds:
\begin{equation}\label{euk-f-int}
\int_{\CC^{0|n}} \exp\Bigl( -\frac 12 (c,Bc)+P(c)\Bigr)\, dc =
 \operatorname{Pf}(-B)\sum_\Ga(-1)^{c(D(\Ga))}\frac{F(D(\Ga))}{|\operatorname{Aut}(\Ga)|},
\end{equation}
where the summation is taken over finite graphs,
$D(\Ga)$  is a mapping of $\Ga$ to $\RR^2$, with the only singular points being crossings of edges ($D(\Ga)$ is a diagram of the graph $\Ga$), and $c(D(\Ga))$ is the number of crossings of edges in the diagram $D(\Ga)$. The number $F(D(\Ga))$ is computed by the same rules as in the previous section. The product $(-1)^{c(D(\Ga))}F(D(\Ga))$ does not depend on the choice of the diagram.
\end{theorem}
\begin{proof}
Expand the integral in $P(c)$:
\begin{multline}
\int_{\CC^{0|n}} \exp\Bigl( \frac 12 (c,Bc)+P(c)\Bigr)\, dc =
\sum_{n_4,n_6,\dots \geq 0}\frac{1}{n_4!(4!)^{n_4}n_6!(6!)^{n_6}\dots}\\
\sum_{i_1,i_2,i_3,\dots}P^{(4)}_{i_1,i_2, i_3,i_4}\dots P^{(6)}_{i_{4n_4+1},i_{4n_4+2},i_{4n_4+3}, i_{4n_4+4},i_{4n_4+5},i_{4n_4+6}}\dots\\
\int_{\CC^{0|n}} \exp\Bigl( \frac 12 (c,Bc)\Bigr)c^{i_1}c^{i_2}c^{i_3}\dots dc.
\end{multline}
Using the identity (\ref{Pf-corr}) we arrive at the
formula
\begin{multline}
\Pf(B)\sum_{n_4,n_6,\dots \geq 0}\frac{1}{n_4!(4!)^{n_4}n_6!(6!)^{n_6}\dots}\\
\sum_{i_1,i_2,i_3,\dots}P^{(4)}_{i_1,i_2, i_3,i_4}\dots P^{(6)}_{i_{4n_4+1},i_{4n_4+2},i_{4n_4+3}, i_{4n_4+4},i_{4n_4+5},i_{4n_4+6}}\dots\\
\sum_m (-1)^m (B^{-1})_{i_{m_1}i_{m_2}} (B^{-1})_{i_{m_3},i_{i_{m_4}}}\dots,
\end{multline}
where $m$ is a perfect matching on $1,2,\dots,k$, $k= \sum_{i\geq 3} i n_i$. The summation over $\{ i\}$ gives the number $F(D_m)$,
where $D_m$ is the diagram from Fig. \ref{diag-match}. Some of the
diagrams $D_m$ represent projections of the same graph. It is easy to show that the combination $(-1)^mF(D_m)$ depends only on the graph, but
not on its diagram and is equal to $(-1)^{c(D(\Ga))}F(D(\Ga))$ for
any diagram $D(\Ga)$ of $\Ga$. Thus, if we will change the summation from
$n_i$ and $m$ to the summation over graphs, the factorials
together with the number of perfect matchings corresponding to the same graph produce the
combinatorial factor $1/|\Aut(\Ga)|$.
\end{proof}

Having in mind applications to oscillatory integrals, it is
convenient to have the
formula (\ref{euk-f-int}) in the form
\begin{equation}\label{euk-f-intv2}
\int_{\CC^{0|n}} \exp\Bigl( \frac{i}{2h} (c,Bc)-\frac{i}{h}P(c)\Bigr)\, dc =h^{-\frac{n}{2}}
 \Pf(iB)\sum_\Ga(ih)^{-\chi(\Ga)}(-1)^{c(D(\Ga))}
 \frac{F(D(\Ga))}{|\Aut(\Ga)|}.
\end{equation}

\subsection{Formal asymptotics of oscillatory integrals
over supermanifolds} \index{oscillatory integrals over
supermanifolds} There is a number of equivalent definitions of
super-manifolds. For our goals a super-manifold $M_{(n|m)}$ is a
trivial vector bundle over a smooth $n$-dimensional manifold $M$
(even part) with the fiber which is the exterior algebra of an
$m$-dimensional vector space $V$ (odd part). The algebra of
functions on such a super-manifold is the algebra of sections of
this vector bundle with the point-wise exterior multiplication on
fibers, i.e. if $f,g: M\to M\times \wedge V$ are two sections
$x\mapsto (x, f(x))$ and $x\mapsto (x,g(x))$, their product is the
section
\[
x\mapsto(x,f(x)\wedge g(x)).
\]
In other words, it is the tensor product of the Grassmann algebra of
the fibers with the algebra of smooth functions on $M$, i.e.
\[ C^\infty(M_{(n|m)})=C^\infty(M)\otimes_\RR\left< c^1,\dots, c^m|c^\alpha
c^\beta=-c^\beta c^\alpha\right> .\] Elements of the algebra are polynomials
in anti-commuting variables $c^1,\cdots, c^m$ with coefficients in
smooth functions on $M$:
\begin{equation}\label{s-fncn}
f(x,c)=f_0(x)+\sum_{k=1}^m \sum_{\alpha_1<\dots <\alpha_k}f_{\alpha_1,\dots, \alpha_k}(x) c^{\alpha_1}\dots c^{\alpha_k}.
\end{equation}

Let $dx$ be a volume form for the manifold $M$. Choose the
orientation $c_1\dots c_m$ on the fibers.
By definition, the integral of the function $f(x,c)$ with respect
to the volume form $dx$ and the orientation $c_1\dots c_m$ is
\[
\int_{M_{(n|m)}}f dx dc=\int_M f_{1,\dots, m}(x) dx.
\]

An even function on such a super-manifold has only terms of
even degree in (\ref{s-fncn}).
{\it Critical points of an even function} $f$ on $M_{(n|m)}$ are,
by definition, critical points of $f_0$ on $M$.

Let $f$ be an even function on $M_{(n|m)}$. Consider the following
integral
\begin{equation}\label{s-int}
\int_{M_{(n|m)}}\exp(\frac{if(x,c)}{h}) g(x,c) dx dc.
\end{equation}
Here we assume that $M$ is compact, and that all functions are smooth.

Combining asymptotic analysis and the asymptotic expansion for
oscillating integrals with the formulae for Grassmann integrals
obtained in the previous section we arrive at the following
asymptotic expansion for the integral (\ref{s-int})
\begin{multline}\label{s-int-F}
\int_{M_{(n|m)}} \exp(\frac{if(x,c)}{h}) g(x,c) dx dc\simeq h^{\frac{n-m}{2}}(2\pi)^{\frac{n}{2}}\\
\sum_a\frac{1}{\sqrt{|\det(B(a))|}}\Pf(iL(a))
\exp({\frac{i}{h}f(a)+\frac{i\pi}{4}\sign(B(a))})\\
\left(1+ \sum_{\Ga\neq \O}\frac{(ih)^{-\chi(\Ga)}(-1)^{c(D(\Gamma))}F_a(D(\Ga))}{|\Aut (\Ga)|}\right).
\end{multline}
Here $B(a)_{ij}=\frac{\pa^2 f}{\pa x^i \pa x^j }(a)$ and
$L(a)_{\alpha\beta}=f_{\alpha\beta}(a)$, the summation is over
finite graphs with  two types of edges: fermionic edges (dashed),
and bosonic edges (solid), $c(D(\Ga))$ is the number of crossings of
fermionic edges in the diagram. Weights of edges (propagators) and
of vertices are given in Fig. \ref{super-weights}. An example is given
in Fig. \ref{super-diag}.

\begin{figure}[htb]
\sidecaption
\includegraphics[height=4cm,width=7cm]{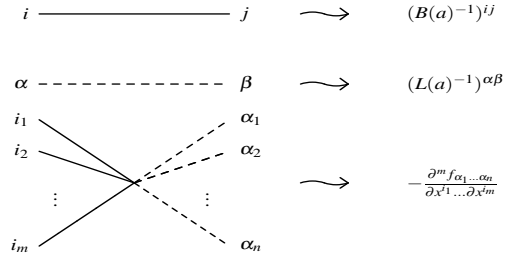}
\caption{Weights for Feynman diagrams in (\ref{s-int-F})}
\label{super-weights}
\end{figure}

\begin{figure}[htb]
\sidecaption[t]
\includegraphics[height=4cm,width=5cm]{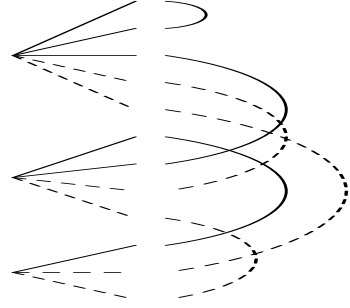}
\caption{An example of the Feynman diagram for super-integrals}
\label{super-diag}
\end{figure}

\subsection{Charged fermions}
Assume that $m=2k$. Denote $c^i=c^i, \bar{c}^i=c^{k+i}$ for $i=1,\dots, k$.  Assume that the
function $f$ in (\ref{s-int}) has the form
\[
f(x,c,\bar{c})=f_0(x)+\sum_{\alpha,\beta=1}^k f_{\alpha\bar{\beta}}(x)c^\alpha \bar{c}^{\bar{\beta}}+\dots\, ,
\]
where $\dots $ denote terms of higher order in $c, \bar{c}$.

In this case the asymptotic expansion of the integral (\ref{s-int})
is given by Feynman diagrams with {\it oriented} fermionic edges:

\begin{multline}\label{s-int-ch-F}
\int_{M_{(n|2k)}} \exp(\frac{if(x,c)}{h}) g(x,c) dx dc=h^{\frac{n-2k}{2}}(2\pi)^{\frac{n}{2}}\\
\sum_a\frac{1}{\sqrt{|\det(B(a))|}}\det(L(a))
\exp(\frac{i}{h}f(a)+\frac{i\pi}{4}\sign(B(a)))\\
\left(1+ \sum_{\Ga\neq \O}\frac{(ih)^{-\chi(\Ga)}(-1)^{c(D(\Gamma))}F_a(D(\Ga))}{|\Aut (\Ga)|}\right),
\end{multline}
where all ingredients are the same as in (\ref{s-int-F}) except the
summation is taken over the graphs with oriented fermionic edges,
and with weights from Fig. \ref{weight-orient}. An example is
given in Fig. \ref{F-diag-orient}.

\begin{figure}[htb]
\sidecaption
\includegraphics[height=4cm,width=7cm]{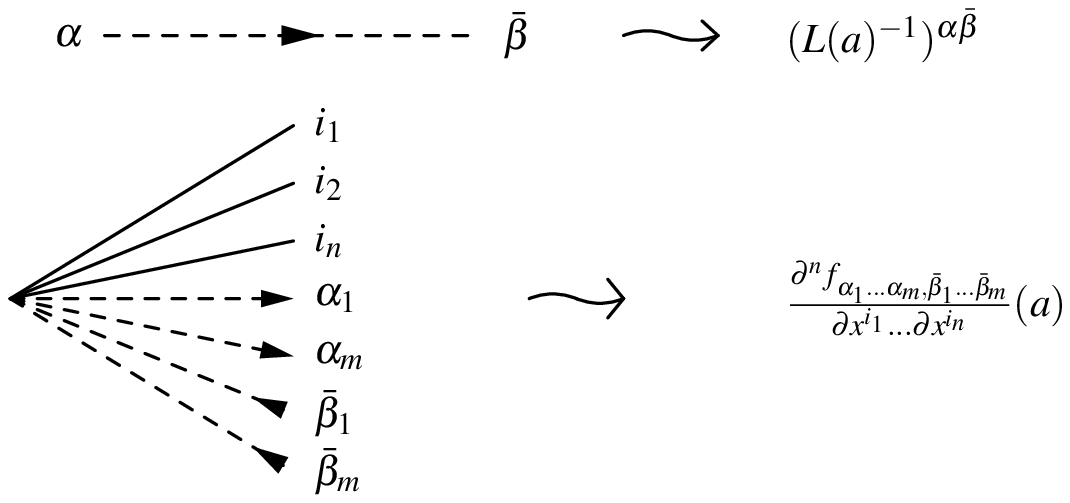}
\caption{Weights for Feynman diagrams in (\ref{s-int-ch-F})}
\label{weight-orient}
\end{figure}

\begin{figure}[htb]
\sidecaption[t]
\includegraphics[height=4cm,width=5cm]{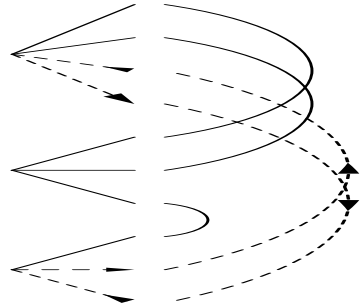}
\caption{Pairing with oriented edges,
producing a Feynman diagram}
\label{F-diag-orient}
\end{figure}

\section{Finite-dimensional Faddeev-Popov quantization and the BRST differential}
\index{Faddeev-Popov} In this section we will study the integral
(\ref{int-f}) when a Lie group $G$ acts on $X$ faithfully (with no
stabilizers) and the function $f$ is invariant with respect to this
action.

\subsection{Faddeev-Popov trick} Let $X$ be a manifold with the action of a  Lie group
$G$. We assume here that the action is free, i.e. that the stabilizer of every point in $X$ is trivial. Assume also that
$X/G$ is a manifold. (Note that what is really important is
the assumption that $X/G$ is smooth near orbits where $f$ is
critical). In this case
\[
\dim(X/G)=\dim(X)-\dim(G).
\]
Assume that the manifold $X$ has a $G$-invariant volume form, and that $X$ is compact. It is clear that such restrictions are
too strong, but we will see in the next section how they can be relaxed to reasonable assumptions.

Let $f(x)$ be a $G$-invariant real analytic function and $dx$ be a $G$-invariant measure on $X$. The goal of
this section is to prepare the set-up for the description of the
asymptotic expansion of the integral
\begin{equation}\label{IFP}
I_h=\int_{X} \exp\left(i\frac{f(x)}{h}\right)dx
\end{equation}
as the sum of Feynman diagrams, just as it was done in section for
functions on $X$ with simple critical points.

Since the function $f$ is $G$-invariant, its critical points are not simple, except when a critical point is a fixed point of the $G$-action, but since we assume faithfulness, there are no
such points.

Instead of assuming the simplicity of critical points of $f$ we assume that critical variety $C_f=\{x\in X|df(x)=0\}$ of $f$ is the disjoint union of finitely many $G$-orbits.

We want to change the integration
over $X$ to the integration over the orbits of the $G$-action.
In practice, it is convenient to describe the space of orbits in terms of a cross-section.

Let us assume that the surface
\[
S_\phi=\{x\in X| \phi^a(x)=0, a=1,\dots, n\},
\]
where $\phi^a(x), a=1,\dots, n$ with $n=\dim(G)$ are independent functions, is a cross-section, i.e. intersects every $G$-orbit
exactly once.

Let $x^i, i=1,\dots, d$ be local coordinates on $U\subset X$, $e_a, a=1,\dots, n$ be a linear basis in the Lie algebra $\g=Lie(G)$.
Denote by $D_a^i(x)$ the matrix
describing the action of $e_a$ as a vector field on $X$ in terms of local coordinates $x^i$:
\[
(e_a f)(x)=\sum_{i=1}^dD_a^i(x)\frac{\p f}{\p x^i}(x)
\]
and by $L_c^b(x)$ the matrix:
\[
L_a^b(x)=\sum_{i=1}^d D_a^i(x)\frac{\p \phi^b}{\p x^i}(x)=e_a\phi^b(x).
\]
Since we assume that $S_\phi$ is a cross-section, $\det(L)\neq 0$ on this surface. Later we will relax this condition requiring only that the determinant is not vanishing in a vicinity of critical points of $f$.

A coordinate free way to formulate this non-degeneracy condition can be phrased as follows.
For $x\in S_\phi\subset X$ let  $L_x\subset T_xX$ be the subspace of the tangent space
spanned by vector field describing the action of the Lie algebra
$\g$ on $X$, and $T_xS_\phi\subset T_xX$ be the tangent
space to $S_\phi$ at this point. The non-degeneracy of $L$ is
equivalent to linear independence of $L_x$ and $T_xS_\phi$ in $T_xX$.

\begin{theorem}(Faddeev-Popov)\footnote{Faddeev and Popov derived the formula (\ref{FP}) in the setting of the Yang-Mills theory,
 where the symmetry group is infinite-dimensional and only the integration over gauge classes may have a meaning, see \cite{FP}. } The integral in question is given by
\begin{equation}\label{FP}
\int_{X} \exp\left(i\frac{f(x)}{h}\right)dx=h^n|G|\int_L \exp\left(i\frac{f_{FP}(x,c,\bar{c},\lambda)}{h}\right)dx \ d\bar{c} \ dc \ d\lambda,
\end{equation}
where the supermanifold $L$ is $X\times \g_{odd}\times \g_{odd}\times \g^*$, $|G|$ is the volume of $G$ with respect to
a left invariant measure $dg$,
and
\begin{equation}\label{FPA}
f_{FP}(x,c,\bar{c},\lambda)=f(x)-ih\sum_{a,b=1}^n c^a L_a^b(x) \bc_b+\sum_{a=1}^n\lambda_a\phi^b(x).
\end{equation}
\end{theorem}

\begin{proof} From the $G$-invariance of $f$:
\begin{equation}\label{FP-int-proof}
\int_{X}\exp\left(i\frac{f(x)}{h}\right)dx=\int_{X\times G}
\exp(i\frac{f(x)}{h}) \Delta(x)\delta(\phi(x))dxdg,
\end{equation}
where $\Delta(x)$ is determined by the identity
\begin{equation}\label{Delta}
\Delta(x)\int_G\delta(\phi(gx))dg=1.
\end{equation}
Here $dg$ is the right-invariant measure on $G$.
The $G$-orbit through $x$ intersects the cross-section $S_\phi$
only once (since it is a cross-section). Denote this point $g_0x$ (such $g_0$ depends on $x$, it exists because $S_\phi$ is a cross-section,
and, in particular, intersects all orbits). Then, we have
\[
\phi^a(g_0x)=0.
\]
In a vicinity of this point
\[
\phi^a(\exp(\sum_b t^be_b)g_0x)=\sum_{b,i} t^bD_b^i(g_0x)\frac{\p\phi^a(g_0x)}{\p x^i}+O(t^2)=\sum_b
t^bL_b^a(g_0x)+O(t^2).
\]
Thus, the identity (\ref{Delta})
is equivalent to
\[
\Delta(x)\int_{\RR^n}\delta(L(g_0x)t)dt=1,
\]
i.e.
\[
\Delta(x)=\det(L(g_0x)).
\]
Here we identified $T_{g_0}G\simeq \RR^n$. Notice that
$g_0$ depends on $x$ and $\Delta(hx)=\Delta(x)$.
Taking this into account we arrive at the formula
\[
\int_{X}\exp\left(i\frac{f(x)}{h}\right)dx=
|G|\int_{X}\exp\left(i\frac{f(x)}{h}\right)
\det(L(g_0x))\delta(\phi(g_0x))dx.
\]
Taking into account that $g_0=1$ when $x\in S_\phi$, we obtain
\[
\int_{X}\exp\left(i\frac{f(x)}{h}\right)dx=
|G|\int_{X}\exp\left(i\frac{f(x)}{h}\right)
\det(L(x))\delta(\phi(x))dx.
\]
Expressing $\det(L(x))$ as a fermionic integral and taking into account
\[
\delta(\phi)=\int_{\RR^n} \exp(i(\phi,\lambda)) d\lambda,
\]
we arrive at the formula (\ref{FP}).

\end{proof}

\subsection{Feynman diagrams with ghost fermions} Now let us use the formula (\ref{FP}) to derive the
Feynman diagram expansion of the integral (\ref{IFP}).

Critical points of the function $f_{FP}$ on the supermanifold $L$ are, by definition,  critical points of
\[
\tilde{f}(x,\lambda)=f(x)+\sum_a\lambda_a \phi^a(x).
\]
This is simply the Lagrange multiplier method and by the
assumption which we made above critical points of this function on
$X\times \g^*$ are simple. In particular the matrix of second derivatives is non-degenerate near each critical point of this function.

Thus, we can describe the asymptotic expansion of the integral
(\ref{IFP}) by Feynman diagrams. Applying the formula
(\ref{s-int-ch-F}) to the integral (\ref{FP}) we obtain the
following asymptotic expansion:
\begin{multline}\label{as-FP}
\int_L \exp\left(\frac{if_{FP}(x,\overline{c},c,\lambda)}{h}\right) g(x,c,\bar{c}) dx dc d\overline{c} d\lambda \simeq |G|h^{\frac{d-n}{2}}(2\pi)^{\frac{d+n}{2}}\\
\sum_a\frac{1}{\sqrt{|\det(B(a))|}}\det(-iL(a))
\exp\left(\frac{i}{h}f(a)+\frac{i\pi}{4}\sign(B(a))\right)\\
\left(1+ \sum_{\Ga\neq \O}\frac{(ih)^{-\chi(\Ga)}(-1)^{c(D(\Gamma))}F_a(D(\Ga))}{|\Aut (\Ga)|}\right),
\end{multline}
Here the first summation is over the set of critical points of
$\tilde{f}$. Feynman diagrams in this formula have bosonic edges and
fermionic oriented edges, $c(D(\Ga))$ is the number of crossings of
fermionic edges. The structure of Feynman diagrams is the same as in
(\ref{s-int-ch-F}). The propagators corresponding to Bose and Fermi
edges are shown in Fig. \ref{FP-diag-edges}. The weights of vertices  are
shown in Fig. \ref{FP-diag-vert}\footnote{Each fermionic propagator
contributes to the weight of the diagram an extra factor $h^{-1}$.
Each vertex with two adjacent fermionic (dashed) edges contributes
the factor of $h$. Because fermionic lines form loops, these factors
cancel each other.}.

\begin{figure}[htb]
\includegraphics[height=2.5cm,width=8cm]{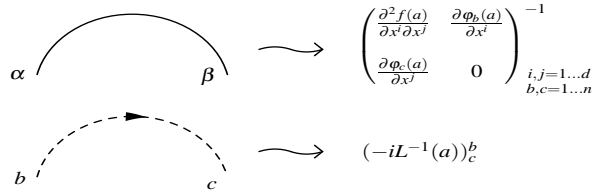}
\caption{Weights of edges for Feynman diagrams in (\ref{as-FP})}
\label{FP-diag-edges}
\end{figure}

\begin{figure}[htb]
\sidecaption[t]
\includegraphics[height=4cm,width=7cm]{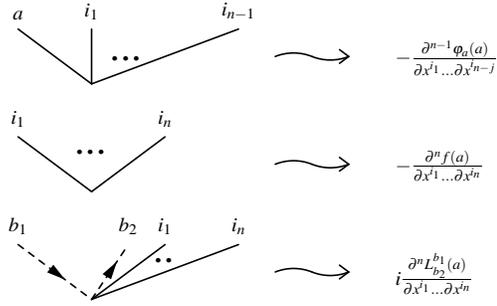}
\caption{Weights of vertices for Feynman diagrams in (\ref{as-FP})}
\label{FP-diag-vert}
\end{figure}

The asymptotic expansion (\ref{as-FP}) depends only on how the cross
section $S_\phi$ intersects $G$-orbits in the infinitesimal
neighborhood of critical points of $f$. In other words, the
expansion is defined as long as the linear operators $B(a)$ and $L(a)$
are invertible at all critical points of the function
$\tilde{f}(x,\lambda)$. This is equivalent to the condition
$T_aS_\phi\cap \g_a=\{0\}$ where $T_aS_\phi\subset T_aX$ is the
tangent space to $S_\phi$ at $a$, and $\g_a$ is the subspace in
$T_aX$ spanned by vector fields describing the infinitesimal action
of the Lie algebra of $G$.

The main moral of this observation is that in order to have the
asymptotic expansion of the integral in terms of Feynman diagrams we just have to choose a constraint which is a cross-section through the orbits in an infinitesimal neighborhood of critical orbits.

\subsection{Gauge independence}\index{gauge independence}   The asymptotic expansion
of the integral (\ref{IFP}) does not depend on the choice of the
constraint $\phi$ (as long as it is  a cross-section through the
$G$-orbits of tangent spaces at critical points).

However, it is not obvious from the Feynman diagram formula for the
asymptotic expansion. Let us check that the semiclassical term of
the expansion does not depend on $\phi$. Till the end of this section we work in a vicinity of a critical point of $f_{FP}$.  The semiclassical term is
\[
\det(B)^{-\frac{1}{2}}\det(L),
\]
where
\begin{equation}\label{B}
B=\left(\begin{array}{cc} \frac{\pa^2 f}{\pa x^i \pa x^j} & \frac{\pa \phi^b}{\pa x^i}\\
\frac{\pa \phi^a}{\pa x^j} & 0
\end{array}\right)
\end{equation}
and
\begin{equation}\label{L}
L_c^b=\sum_i l^i_c\frac{\pa \phi^b}{\pa x^i}.
\end{equation}
Let us make an infinitesimal variation of the constraint $\phi^a(x)\mapsto \phi^a(x)+\ep^a(x)$. The product of the determinants
will change as
\[
\det(B)^{-\frac{1}{2}}\det(L)\mapsto\det(B)^{-\frac{1}{2}}\det(L)\left(1-\frac{1}{2}\tr (B^{-1}\delta B)+ \tr (L^{-1}\delta L)+\dots\right),
\]
where $\dots$ are higher order terms.
We have to prove that the first order terms vanish.
The matrix $B$ has the block form, so is the matrix $B^{-1}$.
Both of these matrices are symmetric, therefore
\[
-\frac{1}{2}\tr (B^{-1}\delta B)=-\tr ((B^{-1})_{12}\delta B_{21})=
-\sum_{i,c}b_c^i \frac{\pa \epsilon^c}{\pa x^i},
\]
where $b_a^i$ are matrix elements of the block $(B^{-1})_{12}$.
They satisfy the
identity $\sum_i  \frac{\pa \phi^b}{\pa x^i}b_c^i=\delta_c^b$.

The second term can be written as
\[
\tr (L^{-1}\delta L)=\sum_{b,c,i}(L^{-1})^c_b l^i_c\frac{\pa \epsilon^b}{\pa x^i}.
\]
Using the identity $\sum_i(L^{-1}l)^i_b\frac{\pa \phi^c}{\pa x^i}=\delta^c_b$ and the corresponding identity for $b$
we conclude that
\[
-\frac{1}{2}\tr (B^{-1}\delta B)+\tr (L^{-1}\delta L)=0,
\]
which proves that the semiclassical factor does not depend on the
choice of the gauge condition.

We will leave the exercise
of verifying this fact in all orders $\geq 1$ to the reader.

\subsection{Feynman diagrams for linear constraints}\label{FP-gost-restr} Because the asymptotic expansion depends only on the formal neighborhood of critical points of $f(x)$ on the
surface of the constrains and does not depend on the
particular choice of the constraint (as long as it is
a local cross-section in the neighborhood of each critical
point), we can choose them at our convenience at each neighborhood.

In particular, if $X$ is linear, we can deform $\phi$ to a linear
cross-section in a formal neighborhood of each critical point. Now let us find the
asymptotic expansion of the integral
\begin{equation}\label{int-lin}
\int_{X_{x_0}}\exp\left(i\frac{f(x)}{h}\right)
\det(L(x))\delta(\phi(x))dx,
\end{equation}
where $X_{x_0}$ is an infinitesimal neighborhood of $a\in X$.

Choose coordinates $x^i=x_0^i+\sum_al_a^i(x_0)X^a+\sum_\alpha\psi^i_\alpha s^\alpha$
where $\{\psi_\alpha\}$ is a basis in $ker(\phi)$. Taking into account that $x_0\in ker(\phi)$, for the integral
(\ref{int-lin}) we obtain
\[
\int_{X_a}\exp\left(i\frac{f(x)}{h}\right)
\det(L(x))\delta(\phi(x))dx=\int_{ker(\phi)}\exp\left(i\frac{f(s)}{h}\right)\det(L(s))\det(\hat{\phi})^{-1}ds,
\]
Since the constraint is linear, $L_a^b(s)=\sum_i l^i_a(x) \phi^a_i$.
Here $l_a^c(s)$ is the matrix describing the action of
$\g$ on $\g(a) \subset T_aX$.

Thus, up to a constant, the integral in question can be written as
\[
\int_{X_\phi(a)}\exp\left(i\frac{f(s)}{h}\right)\det(L(s))ds.
\]

\begin{figure}[htb]
\sidecaption
\includegraphics[height=4cm,width=7cm]{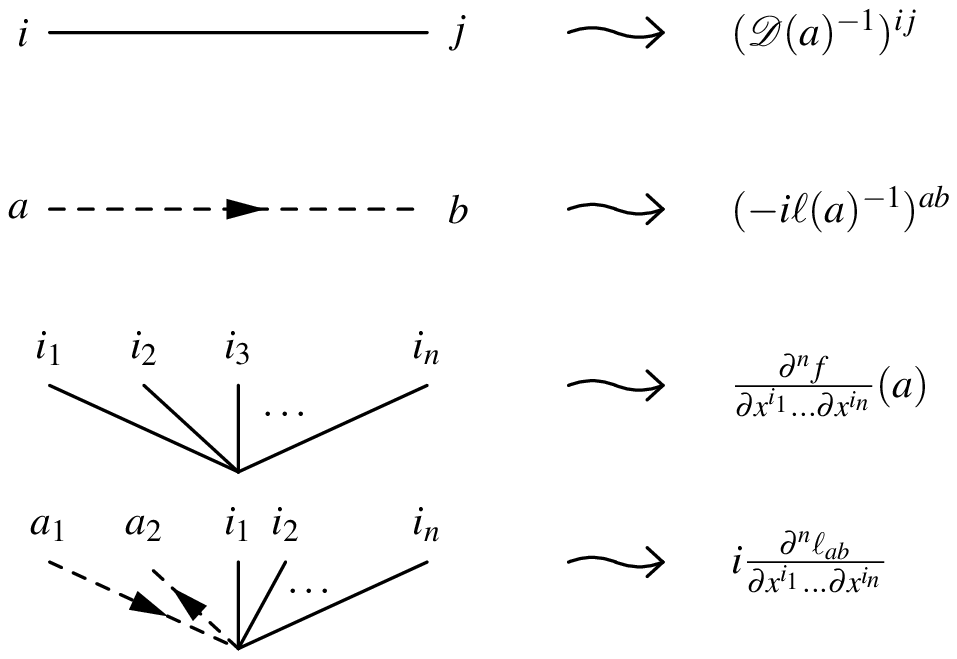}
\caption{Weights of Feynman diagrams in (\ref{red-FP})}
\label{reduced-FP-diag}
\end{figure}

Finally, we can write the asymptotic expansion of (\ref{FP}) as
\begin{multline}\label{red-FP}
\int_{X}\exp\left(i\frac{f(x)}{h}\right)dx\simeq |G|h^{\frac{d-n}{2}}
(2\pi)^{\frac{d+n}{2}}\\
\sum_a \frac{1}{\sqrt{|\det(D(a))|}}\det(-il(a))
\exp({\frac{i}{h}f(a)+\frac{i\pi}{4}\sign(D(a))})\\
\left(1+ \sum_{\Ga\neq \O}\frac{(ih)^{-\chi(\Ga)}(-1)^{c(D(\Gamma))}F_a(D(\Ga))}{|\Aut (\Ga)|}\right).
\end{multline}
Here $D(a)_{ij}=\frac{\pa^2 f}{\pa s^i \pa s^j}$ where $s^i$
are coordinates on $X_\phi$. The coefficients are given by Feynman diagrams with
weights of edges and vertices described in Fig. \ref{reduced-FP-diag}, and all other ingredients
are as before.

The factor $\exp(\frac{i\pi}{4}\sign(B))$ can also be written as
$i^N\exp(-\frac{i\pi}{2}n_-(B))$ where $n_-(B)$ is the number of
negative eigenvalues of $B$. This is more or less how the {\it Morse
index} appears in the semiclassical asymptotic of the propagator in
quantum mechanics.

\subsection{The BRST differential}
\index{BRST differential} The appearance of fermionic variables
(Faddeev-Popov ghost fields) in the asymptotic expansion of
(\ref{FP}) looks as a bit of a mystery and as a technical trick. In
the BRST approach these non-commutative variables attain a natural
meaning.

The key observation of Becchi, Rouet, Stora \cite{BRS} and of Tuytin \cite{T}\footnote{The original formulation uses the supersymmetry
concept and has a slightly different appearance.} is that
the odd operator $Q$
\[
Q=\sum_{a,i=1}^{n,d} c^a L_a^i\frac{\p}{\p x^i}-\frac{1}{2}\sum_{a,b,c} f_{bc}^ac^bc^c\frac{\p}{\p c^a}+\sum_a\lambda_a \frac{\p}{\bc_a}
\]
acting on the space $C^\infty(L)=Fun(X\times \g^*)\otimes \CC[c^a,\bc_a]=C^\infty(X\times \g^*)\otimes \wedge(\g\oplus \g^*)$
of functions on the super-manifold $L=X\times \g_{odd}\times \g_{odd} \times \g^*$ has the properties
$$Q^2=0,$$
$$Qf_{FP}=0.$$

The first property means that the pair $(C^\infty(L),Q)$ is a co-chain
complex. Because we assumed that the action of $G$ on $X$
is faithful, its zero cohomology can be naturally identified with $C^\infty(X/G)$, and other the cohomologies vanish.
Note that $Q=Q_{Ch}+Q_{dR}$, where the first term is the differential in the  Chevalley complex for $\g$ with the coefficients in $C^\infty(X)$. The second term $Q_{dR}=\sum_a\lambda_a \frac{\p}{\bc_a}$ is the de Rham differential for functions on $\g^*$.

The second property means that
the Faddeev-Popov action is a cocycle in the BRST complex.
The function $f_{FP}$ is not a co-boundary and therefore defines
a non-trivial zero-cohomology class in $H^0(L)\simeq C^\infty(X/G)$.
This class is simply the initial function $f$ considered as
a function on $G$-orbits. Indeed, the function $f_{FP}$ can be written as $f_{FP}=f+Q(\sum_a \phi^a \bc_a)$.

To see how the integral over the
super-space $L$ appears in this setting, consider first
a simple fact in linear algebra.

Let $C$ be a super-vector space and $d: C\to C$ be an odd linear
operator with $d^2=0$. Assume $D$ is another super-vector space with an odd differential $d^*:D\to D, {d^*}^2=0$ and a non-degenerate pairing $\left< .,.\right> : D\otimes C\to \CC$ such that
$\left< d^*l,a\right> =(-1)^{\bar{l}}\left< l, da\right> $.

We will think of $(D, d^*)$ and $(C,d)$ as co-chain complexes and
say that $l\in D$ and $a\in C$ are cocycles, if $d^*l=0$ and $da=0$.
Denote by $[l]\in H(D)=\Ker(d^*)/\im(d^*)$ and $[a]\in
H(C)=\Ker(d)/\im(d)$ the cohomology classes of the cocycles $l$ and
$a$.

\begin{lemma} If  $l$ and $a$ are cocycles, then
\[
\left< l,a\right> =\left< [l],[a]\right> ,
\]
where $\left< [l],a]\right> $ is the induced pairing on the
cohomology spaces.
\end{lemma}
Indeed, the cocycle properties imply that
\[
\left< l+d^*m,a+dc\right> =\left< l,a\right> ,
\]
which defines the pairing on the cohomology spaces and proves the lemma.

Now we should identify ingredients of this lemma in the
FP-BRST setting.
The $G$-invariance of the measure of integration $dxd\bar{c}dcd\lambda$ in the
FP integral which we will denote by $dl$ implies\footnote{ The operator $Q$ can be regarded as an super-vector field on $L$. The invariance of the measure $dl$ is equivalent to the zero-divergence condition of the vector-field (with respect to the measure $dl$). Recall that for any vector field $Q$ we have:
\[
\int_LQg dl=\int_L g \Div_{dl}(Q) dl
\]
where $\Div_{dl}(Q)$ is the divergence of the vector field $Q$ with
respect to the volume measure $dl$.}
\[
\int_L Qg dl=0
\]
Considering the integral as a linear functional on $C=C^\infty(L)$
with the differential we can think of it as an element of $D$ which is annihilated by $d^*$.

Applying the lemma to a cocycle $g\in C^\infty(L)$, i.e.
to a function, such that $Qg=0$ we arrive at the identity
\begin{equation}\label{lefsh}
\int_L g dl= \int_Y [g] dy.
\end{equation}
Here $Y$ is the super-manifold such that $H^0(C^\infty(L))=
C^\infty(Y)$, i.e. the appropriate topological version of
$X/G$. If we were in algebro-geometric setting,
the variety $Y$ would be the spectrum of the commutative algebra
$H^0$. We also made an assumption that all cohomologies except $H^0$ are vanishing, which is in our setting equivalent to the faithfulness of the $G$-action on $X$.

The equation (\ref{lefsh}) implies, in particular, that
if $Qg=0$ (i.e. if $g$  is $G$-invariant) and if the measure is
$G$-invariant, then
\[
\int_L \exp\left(\frac{if_{FP}}{h}\right)g dl= \int_Y \exp\left(\frac{if}{h}\right)[g] dy.
\]

This puts the Faddeev-Popov method into a natural algebraic setting and
''explains`` the algebraic meaning of fermionic ghost fields. It also
shows that the method can be extended to any complex which has
$C^\infty(X/G)$ as its cohomology space. Because of the formula
(\ref{lefsh}) it does not matter with which complex
$(C^\infty(L),Q)$ we started, as long as the cohomology space is
$C^\infty(X/G)$. This observation leads to an important notion of
cohomological field theories \cite{CohFTh} and to a natural notion
of quasi-isomorphic field theories.

Perhaps one of the most important developments along these lines is the extension of the BRST observations to a more general
class of degenerate Lagrangians (i.e. degenerate critical points of $f$). This generalization known as Batalin-Vilkovisky quantization (BV) works even in the case
when the Lagrangian is invariant with respect to
the action of vector fields, which do not necessary form
a Lie algebra. One of the most striking applications of this technique was the quantization of the Poisson sigma model and
the construction of the *-product for an arbitrary Poisson
manifold. But this subject goes beyond the goal of the present
lectures.

\section{Semiclassical quantization of a scalar Bose field}

The classical theory of a scalar Bose field is described in Section
\ref{cBose} Let us define the amplitude $Z(M)$ as a semiclassical
expansion of a (non-existing) path integral given by Feynman
diagrams similar to how the asymptotic expansion looks for finite-dimensional integrals.

This definition can be motivated by finite-dimensional approximations to the path integral, which provide
an acceptable definition of {\it infinite-dimensional integrals}
such as Wiener integral and path integrals in low-dimensional
Euclidean quantum field theories \cite{GJ}.

In semiclassical quantum field theory, path integrals are defined
as formal power series which have the same structure as if they were
asymptotical expansions of existing integrals. The coefficients in
these expansions are given by Feynman integrals. We will show how it
works in quantum mechanics, and how it compares with the
semiclassical analysis of the Schr{\"o}dinger equation for $d=1$. We will have a brief discussion of th e$d>1$ case, as well.

\subsection{Formal semiclassical quantum mechanics}
\index{semiclassical quantum mechanics}
\subsubsection{Semiclassical asymptotics from the Schr{\"o}dinger equation}
To be specific, we will consider here quantum mechanics of a point particle on a Riemannian manifold $N$ in a potential $V(q)$(see sections \ref{cMech}).

Let $\{\gamma_c(t)\}_{t_1}^{t_2}$ be a solution to the Euler-Lagrange equations for a classical Lagrangian $\LL(\xi,q)$
with Dirichlet boundary conditions $\gamma(t_1)=q_1, \gamma(t_2)=q_2$. Denote by $S^{(c)}_{t_2-t_1}(q_2,q_1)$
the value of the classical action functional on $\gamma_c$:
\[
S^{(c)}_{t_2-t_1}(q_2,q_1)=\int_{t_1}^{t_2}\LL(\dot\gamma_c(t),\gamma_c(t))dt.
\]

Let $U_t(q_2,q_1)$ be the kernel of the integral operator
representing the evolution operator (\ref{evol-kernel}). Solving
Schr{\"o}dinger equation (\ref{Sch-eqn}) in the limit $h\to 0$ we
obtain the following asymptotics of the evolution kernel as $h\to
0$

\begin{multline}\label{s-cl-schr}
U_t(q_2,q_1)\sim\sum_{\gamma_c}(2\pi i)^{-\frac{n}{2}}\exp\left(\frac{i}{h}S^{(c)}_t(q_2,q_1)+\frac{i\pi \mu(\gamma_c)}{2}\right)\\
\left|\wedge^n\left(\frac{\pa^2S^{(c)}_t(q_2,q_1)}{\pa q_2\pa q_1} dq_2\wedge dq_1\right)\right|^{\frac{1}{2}}\left(1+\sum_{n\geq 1}h^n U_c^{(n)}(q_2,q_1)\right).
\end{multline}
Here
\begin{multline}
\wedge^n \left(\frac{\pa^2S(a,b)}{\pa a\pa b} da\wedge db\right)\\
=\wedge^n d_{a}d_{b}S(a,b)=
\det\left(\frac{\p^2S(a,b)}{\pa a^i\pa b^j}\right)da^1\wedge\dots da^n\wedge
db^1\wedge\dots db^n\, .
\end{multline}
$\mu(\gamma_c)$ is the Morse index of $\gamma_c$. The coefficients $a_k^{(c)}=(2\pi i)^{-\frac{n}{2}}(\det(\frac{\p^2S(a,b)}{\pa a^i\pa b^j}))^{\frac{1}{2}}U^{(n)}_c$ satisfy the transport equation
\[
\frac{\pa a^{(c)}_k}{\pa t}+\frac{1}{2m}\Delta S^{(c)} a^{(c)}_k+\frac{1}{m}\sum_{j=1}^n\frac{\pa S}{\pa q_j}\frac{\pa a^{(c)}_k}{\pa q_j}+\frac{i}{2m}\Delta a^{(c)}_{k-1}=0.
\]
However the initial condition $\lim_{t\to +0}U_t(q,q')=\delta(q,q')$
can no longer be imposed since we consider the asymptotical
expansion when $h\ll t$. Instead, to determine the coefficients $a^{(c)}_k$, one should use the semi-group property of the propagator:
\[
U_tU_s=U_{s+t}.
\]
The kernel of the integral operators representing the evolution operator satisfies
the identity
\begin{equation}\label{semigroup}
\int_N U_t(q_3,q_2)U_s(q_2,q_1)=U_{s+t}(q_3,q_1).
\end{equation}
Here the first factor is a half-density in $q_3,q_2$, the
second is a half-density in $q_2, q_1$. The product is a
density in $q_2$ and it is integrated over $N$.

As $h\to 0$ the semigroup property implies that the asymptotical
expansion should satisfy the identity
\[
\sum_{k,l\geq 0} \int_N \exp\left(\frac{i(S^{(c')}_t(q_3,q_2)+S^{(c")}_s(q_2,q_1))}{h}\right)
a_k^{(c')}a_l^{(c")} =\sum_k  \exp\left(\frac{iS^{(c)}_t(q_3,q_1)}{h}\right)
a_k^{(c)}.
\]
Here by the integral of the product of two half-densities on $N$ we
mean the formal asymptotic expansion (\ref{as-exp}), and
$\gamma_{c'},\gamma_{c"}$ are parts of the path
$\{\gamma_c\}_{t=0}^{t+s}$ when $0<\tau<s$ and $s<\tau< s+t$
respectively.

It is not difficult to derive the first coefficients of
the asymptotical expansion (\ref{s-cl-schr}) from this
equation. Moreover, this equation alone defines all higher order
terms in the semiclassical expansion.

For more details on the semiclassical analysis see for example \cite{Takht}.

\subsubsection{Semiclassical expansion from the path integral}
\index{path integral} Looking at the expression (\ref{s-cl-schr}),
it is natural to imagine that it may be interpreted as a
semiclassical asymptotics of an oscillating integral over the space
of paths connecting the points $q_1$ and $q_2$. Critical points in
this integral are classical trajectories.

This point of view was put forward in quantum mechanics by R.
Feynman and it can be supported by many very convincing arguments
\cite{FH}. Eventually, a mathematically meaningful definition of a
path integral for the Euclidian version (when the integral is
rapidly decaying instead of oscillating) emerged, and was developed
further in the framework of constructive field theory. The {\it Wiener
integral}, which was introduced in probability theory, even earlier,
is an example of such an object.

Here we will not try to make the definition of the integral
rigorous. Instead of this we will define its semiclassical expansion
in such a way that it has an appearance of a semiclassical
expansion of an infinite-dimensional integral. After this we will
check that it satisfies the semigroup property. This is an
illustration of  a semiclassical quantum field theory, where the
partition function $Z_M$ depends on the boundary condition, and
integrating over possible boundary conditions has the replicating
gluing property (\ref{gluing-prop}). The difference is that in
quantum mechanics we have the Schr{\"o}dinger equation as a
reference point to compare any definition of the path integral. In the
more complicated models of quantum field theory, the gluing axioms
seem to be the only major structural requirement (beyond unitarity
and causality, which we do not discuss here).

So, we are looking for a formal power series which would look like
the asymptotic expansion of the integral
\[
Z_t(q_2,q_1)=\int_{\gamma(0)=q_1, \gamma(t)=q_2}
\exp({\frac{i}{h}S[\gamma]}) D\gamma.
\]
We will focus in this section on the point particle of
mass $m$ in $\RR^n$ in the potential $V(q)$ (\ref{Lpp}).
By analogy with the finite-dimensional case, we define the asymptotic expansion as
\begin{multline}\label{qm-feyn}
Z_t(q_2,q_1)=C\sum_{\gamma_c}\exp\left(\frac{i}{h}S^{(c)}_t(q_2,q_1)
-\frac{i\pi}{2} n_-(K^{(c)})\right)\\
|{\det}'(K^{(c)})|^{-\frac{1}{2}}|dq_1|^{\frac{1}{2}}|dq_2|^{\frac{1}{2}}\left(1+\sum_{\Ga\neq
\emptyset}(ih)^{-\chi(\Ga)} \frac{F_c(\Ga)}{|\Aut(\Ga)|}\right).
\end{multline}
Here
\[
(K^{(c)})_{ij}=-mh^2\frac{d^2}{d\tau^2}\delta_{ij}+\frac{\pa^2 V}{\pa x^i \pa x^j}(\gamma_c(\tau))
\]
is the matrix differential operator which acts on the space of
functions on $[0,t]$ with values in $\RR^n$ (trivialized tangent
bundle to $N$ restricted to $\gamma^{(c)}$ in local coordinates)
with the Dirichlet boundary conditions $f(0)=f(t)=0$. The
half-density $|dq|^{\frac{1}{2}}$ is the ''square root`` of the
Riemannian volume density on $N$. The sum is taken over classical
trajectories connecting $q_1$ and $q_2$, and $n_-(K^{(c)})$ denotes
the number of negative eigenvalues of the operator $K^{(c)}$, $C$ is
some constant. The weights for Feynman diagrams in (\ref{qm-feyn})
are given in Fig. \ref{QM-F-weights}, where $G^{ij}(x,y)$ is the
kernel of the integral operator which is the inverse to $K^{(c)}$.

The expansion is not the result of computation. It is a definition, which is based on the idea that the path integral exists in some
sense and its asymptotical expansion as $h\to 0$ is
given by a formula similar to the finite-dimensional case.
It turns out that despite very different appearance the
semiclassical expansion of $U_t$ coincidea with this series.

\begin{figure}[htb]
\includegraphics[height=3cm,width=8cm]{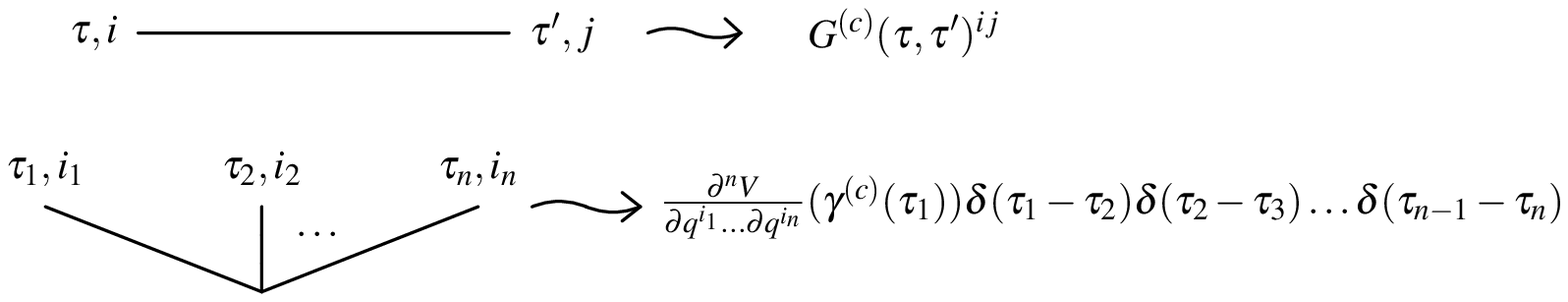}
\caption{Weights of Feynman diagrams in (\ref{qm-feyn})}
\label{QM-F-weights}
\end{figure}

One can show easily (see for example \cite{Takht}) that
\[
|{\det}'(K^{(c)})|=\left|\det\left(\frac{\p^2S(a,b)}{\pa a_i\pa b_j}\right)\right|^{-1},
\]
as well as that $\mu(\gamma^{(c)})$ can be identified with
$n_-(K^{(c)})$. This shows that the leading terms of (\ref{qm-feyn})
and  (\ref{s-cl-schr}) are the same. Now the question is whether the
two power series are the same.

We will state without proof the following theorem.

\begin{theorem} The expansion $Z_t(q_2,q_1)$ is equal to the asymptotic expansion of the kernel of the propagator and it satisfies the gluing property.
\end{theorem}
The details will appear in a paper by T. Johnson-Freyd \cite{JF} when $N=\RR^d$ with flat metric.

As an immediate corollary to this theorem we have
\begin{corollary} The functions
\[
U_c^{(n)}(q_2,q_1))=\sum_{-\chi(\Ga)=n} \frac{F_c(\Ga)}{|\Aut(\Ga)|}
\]
are coefficients of the asymptotical expansion of the
propagator, and, after being properly normalized, satisfy the transport
equation. Here $\chi(\Ga)=|V|-|E|$ is the Euler characteristic.
\end{corollary}

Let us write the semiclassical propagator as
\[
Z_t(q_2,q_1)=\sum_c\exp\left(\frac{i}{h}S^{(c)}_t(q_2,q_1)\right)J^{(c)}_t(q_2,q_1).
\]

The semigroup property of the propagator implies that this
power series satisfies the following gluing/cutting identity:
\begin{multline}\label{gluing-QM}
\exp\left(\frac{i}{h}S^{(c)}_t(q_3,q_1)\right)J^{(c)}_t(q_3,q_1)=\\
\int_{q_2\in N}
\exp\left(\frac{i}{h}\left(S^{(c)}_s(q_3,q_2)+S^{(c)}_{t-s}(q_2,q_1)\right)
\right)J^{(c)}_s(q_3,q_2)J^{(c)}_{t-s}(q_2,q_1).
\end{multline}
Here the integral is taken in a sense of the semiclassical expansion as the sum of corresponding Feynman diagrams.
It is easy to check that the identity (\ref{gluing-QM})
determines uniquely not only the higher order coefficients
but also the leading order factor.

\subsection{$d>1$ and ultraviolet divergencies}
\index{ultraviolet divergencies} In the semiclassical theory of
scalar Bose field on a compact Riemannian manifold the partition
function for the theory is given by the formal power series
\begin{multline}\label{B-f-diag}
Z_M(b)=C\sum_{\phi_c}\exp\left(\frac{iS_M(\phi_c)}{h}-\frac{i\pi}{2}
n_-(K_{\phi_c})\right)|{\det}'(K_{\phi_c})|^{-\frac{1}{2}}\\
\left(1 +\sum_{\Ga\neq \emptyset} (ih)^{-\chi(\Ga)}\frac{F_{\phi_c}(\Ga)}{|\operatorname{Aut}(\Ga)|}\right).
\end{multline}
Here we assume that there are finitely many solutions $\phi_c$ to
the Euler-Lagrange equation (\ref{EUB}) with the Dirichlet boundary
conditions $\phi_{c|_{\pa M}}=b$. The number $n_-(K_{\phi_c})$ denotes the number of
negative eigenvalues of the differential operator
\[
K_{\phi_c}=\Delta+ V''(\phi_c(x))
\]
acting on the space of functions on $M$ with the boundary condition $f(x)=0, x\in \pa M$,
and $\det '(K_{\phi_c})$ is its regularized determinant ($-\frac \pi 2 n_-(K_{\phi_c})$ is
the phase of the square root of the determinant).
The $\zeta$-function regularization is one of the standard ways to define $\det'$ (see for example \cite{AS}).
The weights of Feynman diagrams are given in Fig. \ref{B-Feyn-diag}, where $G^{(c)}(x,y)$ is the kernel of the integral operator
which is inverse to $K_{\phi_c}$.

\begin{figure}[htb]
\sidecaption
\includegraphics[height=2cm,width=8cm]{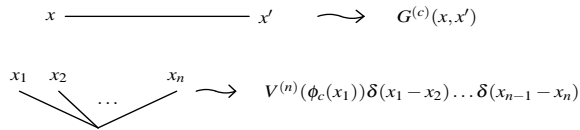}
\caption{Weights of Feynman diagrams in (\ref{B-f-diag})}
\label{B-Feyn-diag}
\end{figure}

An example of an order $1$ Feynman diagram is given in Fig. \ref{B-ord-1}.

\begin{figure}[htb]
\includegraphics[height=1.5cm,width=8cm]{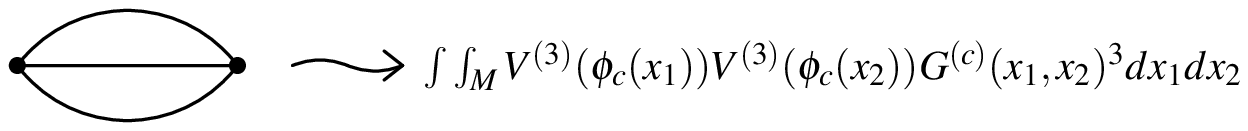}
\caption{An example of the Feynman diagram of order one}
\label{B-ord-1}
\end{figure}

The kernel $G(x,y)$ behaves at the diagonal as
\[
G(x,y)\sim |x-y|^{2-d}\, ,
\]
which means that the Feynman integrals converge
for $d=1$ (quantum mechanics), diverge logarithmically  for
$d=2$, and diverge as a power of the distance for $d>2$.

This is a well known problem of ultraviolet divergencies in the
perturbation theory. The usual way to deal with divergencies is
a two step procedure.

{\bf Step 1}. The theory is replaced by a family of theories
where the Feynman integrals converge (regularized theories).
There are several standard ways to do this:
\begin{itemize}
\item {\it Higher derivative regularization} replaces the theory
with one where the quadratic part of the action
has terms with higher derivatives. In the regularized theory
the propagator $G(x,y)$ is not singular at the diagonal.
For more details see \cite{IZ}.

\item {\it Lattice regularization} replaces the theory on
a smooth Riemannian manifold $M$ by a metrized cell
approximation of $M$. The path integral becomes finite-dimensional and Feynman diagrams describing
the semiclassical expansion become finite sums.

\item {\it Dimensional regularization} is more exotic. It
replaces Feynman $d$-dimensional integrals, where
$d$ is an integer, by formal expressions, where $d$ is
not an integer. It is very convenient computationally
for certain tasks (see for example \cite{CK} and
references therein).

\end{itemize}

{\bf Step 2.} After the theory is replaced by a family
of theories where Feynman integrals converge, one should
compute them and pass to the limit corresponding to
the original theory. Of course the limit will not exist
since some terms will have singularities. In some cases
it is possible to make the parameters in the regularized theory (for
example, coefficients in $V(\phi)$) depend on the parameters of the
regularization in such a way that the sum of Feynman diagrams of
order up to $n$ remains finite when the regularization is removed.
Such theories are called {\it renormalizable in orders up to} $n$.

The compatibility of the gluing/cutting axiom, i.e. an analog of the
identity (\ref{gluing-QM}), and the renormalization is, basically, an
open problem for $d>1$, which requires further investigation. Notice
that for $d=2$ the integration over the boundary fields does not
introduce Feynman diagrams with ultraviolet divergencies, but these
diagrams will diverge for $d>2$. This problem was addressed in the case
of Minkowski flat space-time by K. Symanzik in \cite{Sy}.

\section{The Yang-Mills theory}
\index{Yang-Mills} The classical Yang-Mills theory with Dirichlet
boundary conditions was described in Section \ref{cYM}.

In this section we will define Feynman diagrams for
the Yang-Mills theory following the analogy with the
finite-dimensional case. In these notes we will do it ''half-way``,
leaving the most important part concerned with the ultraviolet
divergencies aside.

Naively, the path integral quantization of the classical $d$-dimensional Yang-Mills theory can be constructed as
follows. Let $G$ be a compact Lie group.
\begin{itemize}
\item To a closed oriented $(d-1)$-dimensional Riemannian manifold
with a principal $G$-bundle $P$ we assign the space of functionals on the space of connections on $P$.
\item To a $d$-dimensional Riemannian manifold $M$ with a
principal $G$-bundle on it, we define the functional $Z$
on the space of connections on $P|_{\pa M}$ as
\[
Z_M(b)=\int_{i^*(A)=b} \exp\left(\frac{i}{h}S_{YM}(A)\right) DA.
\]
\end{itemize}
where $i:\pa M \hookrightarrow M$ is the tautological inclusion
of the boundary and $i^*(A)$ is the pullback of the connection $A$ to the boundary.

Now we can use the Faddeev-Popov Feynman diagrams to {\it define }
the semiclassical expansion of this integral. In the finite-dimensional case, Feynman diagrams were derived as an asymptotic
expansion of the existing integral. To define such expansion, we
should do the gauge fixing and then define the Feynman rules. The Feynman diagrams for the Yang-Mills are
divergent because the propagator is singular at the diagonal
(ultraviolet divergence). Nevertheless, the theory is
{\it renormalizable}, as in the previous example, even better, it is
{\it asymptotically free} \cite{Gr}.
We will not go into the details of the discussion of renormalization
but will make a few remarks at the end of this section.
For more details about quantum Yang-Mills theory and Feynamn diagrams see for example \cite{FS}.
\subsection{ The gauge fixing}
\index{gauge fixing} As we have seen in the finite-dimensional case,
the constraint (gauge fixing) needed to construct the asymptotic
expansion of the integral (\ref{FP}) has to be a cross-section
through the orbits only in the vicinity of critical points (critical
orbits) of the action functional. To define Feynman diagrams for the
Yang-Mills theory, we can follow the same logic. In particular, we
can choose a linear Lorenz gauge condition for connections in the
vicinity of a classical solution $A$.

For a connection $A+\alpha$, where $\alpha$ is a $1$-form (quantum
fluctuation around $A$), the Lorenz gauge \index{Lorenz gauge}
condition is
\begin{equation}\label{gauge}
d_A*\alpha=0,
\end{equation}
where $*$ is the Hodge operation. This condition defines a subspace
in the linear space of $\g$-valued $1$-forms, so we can use the
formula (\ref{red-FP}) which uses no Lagrange multipliers. The
contribution to the path integral from a vicinity of $A$ is then an
''integral`` over the space of $1$-forms $\alpha$ from $Ker(d^*_A)$.

\subsection{ The Faddeev-Popov action and Feynman diagrams}
\index{Faddeev-Popov action} \index{Feynman diagrams} Following
the analogy with the finite-dimensional case define the
Faddeev-Popov action for pure Yang-Mills theory as the following
action with fields $\alpha(x), \bc(x), c(x)$:
\begin{multline}\label{FPYM}
S_A(\alpha)=S_{YM}(A)+\int_M\frac{1}{2}\operatorname{tr}\left< F_A(\alpha),
F_A(\alpha)\right> dx\\- \frac{ih}{2} \int_M *d_A\bc\wedge d_A c-
\frac{ih}{2} \int_M *d_A\bc\wedge [\alpha,c].
\end{multline}
Here $A$ is a background connection which is the
solution to the classical Yang-Mills equations and
$\alpha $ is a $\g$-valued  $1$-form on $M$. The
bosonic part of this action is simply $S_{YM}(A+\alpha)$.

The quadratic part in $\alpha$ and the quadratic part in $\bc, c$ of
the action (\ref{FPYM}) are given by the differential operator
$d_A^*d_A$ which is invertible on the space $Ker(d_A^*)$ with
Dirichlet boundary conditions. Other terms define the weights of
Feynman diagrams. The weights are shown in Fig. \ref{red-YM-FP}.
The functions $G_1^A$ and $G_0^A$ are Green's functions of the
Laplace-Beltrami operator $\Delta=d^*d+dd^*$ on $1$- and $0$-forms
respectively.

\begin{figure}[htb]
\includegraphics[height=5cm,width=9cm]{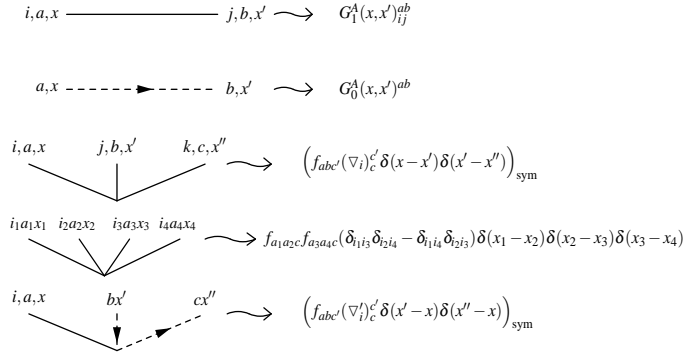}
\caption{Weights of Feynman diagrams in the semiclassical expansion for the Yang-Mills theory}
\label{red-YM-FP}
\end{figure}

\subsection{The renormalization}
The propagator in the Yang-Mills theory is singular at the diagonal
for $d>1$, and just as in the scalar Bose field contributions from
Feynman diagrams to the partition function diverge. However, just as
in the scalar Bose field, when $d\leq 4$, after the renormalization
procedure Feynman diagrams become finite and there is a well defined
semiclassical formal power series for the Yang-Mills theory given by
renormalized diagrams. This fact was discovered by t'Hooft and Veltman \cite{tH} who showed that
using the dimensional regularization of Feynman diagrams and
taking into account Faddeev-Popov ghost fields makes
Yang-Mills into a renormalizable theory.

Moreover, the renormalization in the Yang-Mills theory is remarkable because it gives an asymptotically free theory.
This was discovered in \cite{Gr} and it means that particles in such theory have to behave as non-interacting, free particles at high energies. This {\it prediction} perfectly agrees with {\it experimental data} and this is why the Yang-Mills theory is part of the
Standard Model, unifying the theory of strong, weak, and electromagnetic interactions.

The super-symmetric $N=4$ Yang-Mills theory is expected to have a
particularly remarkable renormalization. It turns out that the divergent
contributions from Feynman diagrams cancel each other in each order
of the expansion in $h$. This was proven in the light-cone gauge
and is believed to be true for other gauges. This Yang-Mills theory
is particularly important for Topological Quantum Field Theories
\cite{KW}\cite{Samson} and in particular to the quantum field
theoretical interpretation of the geometric {\it Langlands program}.

Finally, few words on {\it correlation functions}. Since the
Yang-Mills theory is gauge invariant, natural observables should
also be gauge invariant. Such observables are known as {\it Wilson
loops} \index{Wilson loop}  or, more generally, as {\it Wilson
graphs}.

Recall the definition of Wilson loops.\index{Wilson loop} Let $V$ be
a finite-dimensional representation of a Lie group $G$. The
Yang-Mills potential $A$ (the field in the Yang-Mills theory) is a
connection in a principal $G$-bundle $P$. It induces a connection in
the vector bundle $V_P=P\times_G V$. Let
\begin{equation}\label{W-hol}
h_A(C_x)=P\exp\left(\int_{C_x}A\right)
\end{equation}
be the parallel transport in $V_P$ along a path $C_x$
which starts and
ends at $x\in M$ defined by the connection $A$. Here $P$ stands for the iterated path ordered integral.

The Wilson loop observable is
\begin{equation}\label{W-loop}
W_A^V(C)=\operatorname{Tr}_{V_x}(h_A(C_x)).
\end{equation}
Here the trace is taken over the fiber $V_x$ of $V_P$
over $x\in M$.
The definition of more general gauge invariant
observables, Wilson graphs, will be given later, when
we will discuss observables in the Chern-Simons theory.

An important conjecture about the Yang-Mills theory, and
another fundamental fact expected from this theory, is the
{\it dynamical mass generation}. In terms of expectation values of Wilson
loops, this conjecture means that
\begin{equation}\label{W-asym}
\left< W_A(C)\right> \propto \exp(-ml(C)),
\end{equation}
as $l(C)\to \infty$. Here on the left side we have the expectation value of
the Wilson loop and on the right side $l(C)$ is the length of $C$ in the Riemannian metric on $M$. This conjecture is based on the conjecture that the Yang-Mills theory can be defined non-perturbatively.

The parameter $m$ in (\ref{W-asym}) is characterizing the radius
of correlation. In a massless theory, such as Yang-Mills theory,
there are no reasons to expect that $m\neq 0$.
The appearance of such a parameter with the scaling dimension
of the mass is known as dynamical mass generation.
For more details about this conjecture see \cite{JW}.

\section{The Chern-Simons theory}
\index{Chern-Simons theory} In this section $M$ is a compact
oriented manifold. The classical Chern-Simons theory with a compact
simple Lie group $G$ was described in Section \ref{cCS}. As in the
pure Yang-Mills theory, fields in the Chern-Simons theory are
connections in a principal $G$-bundle over the space-time $M$. In
contrast with the Yang-Mills theory, the Chern-Simons action is the
first order action. One of the implications of this is the
difference in Hamiltonian formulations. The other is that the path
integral quantization for the Chern-Simons theory for manifolds with
boundary is more involved. Some of the aspects of this theory on
manifolds with boundaries can be found in references \cite{Fr},
\cite{ADPW}.

From now on we assume that the space-time $M$ is a  compact, oriented, and closed 3-manifold.
The Chern-Simons action is topological, i.e. its definition does not require a choice of metric on $M$.
This is why it is natural
to expect that the result of quantization, the partition function
$Z(M)$, also depends only on the homeomorphism class of $M$. This gives a powerful criterium for consistency of the definition of Feynman diagrams: the sum of Feynman diagrams for any given order should depend only on the
topology of the manifold.

The path integral formulation of the Chern-Simons theory on manifolds with a boundary is a bit more involved then the
one for the Yang-Mills. This is because the Chern-Simons is a first order theory. The space of states assigned to the boundary is the space of holomorphic sections of the geometric quantization line bundle over the moduli space of flat connections in a trivial principal $G$-bundle over the boundary (provided we made a choice of complex structure).  This space is a quantum
counterpart to the boundary conditions for the Chern-Simons theory when the pull-back to the boundary is required to be holomorphic.
For more details on the quantization of the Chern-Simons theory on
manifolds with boundary see for example \cite{ADPW}.

So, the goal of this section is to make sense of the expression
\begin{equation}\label{CS-pathint}
Z_{M}=\int \exp({ikCS(A)}) DA,
\end{equation}
or, more generally, of
\begin{equation}\label{CS-Wloop}
Z_{M,\Ga}=\int \exp({ikCS(A)}) W_\Ga(A) DA,
\end{equation}
where $W_\Ga(A)$ is a gauge invariant functional
(Wilson graph, or any other gauge invariant
functional) which will be defined later,  and $k$ is an integer which guarantees  that the exponent is gauge invariant. The integral is supposed to be taken
over the space of all connections on a principal $G$-bundle
on $M$.

The integrals (\ref{CS-pathint}),(\ref{CS-Wloop}) are not defined
as mathematical objects. However, one can
try (as in the previous examples of the scalar Bose field and of the
Yang-Mills theory) to define a combination of formal power series in $k^{-1}$ resembling the expansion of finite-dimensional
integrals studied in the previous section. In the case of the Chern-Simons theory, there is a natural requirement
for such expansion: every term should be an invariant of $3$-manifolds.
Remarkably, such a power series exists and is more or
less unique. This program was originated by Witten in \cite{W} who
outlined the basic structure of the expansion. It was developed in a
number of subsequent works, in particular, in
\cite{Kon}, \cite{AS1}, \cite{AS2}, \cite{BC1}, \cite{BC2}, \cite{C} for the
partition function for closed 3-manifolds and in
\cite{DBN}, \cite{GMM}, \cite{AF} for (\ref{CS-Wloop}), and others, when
$\Ga$ is a link.

\subsection{The gauge fixing}
\index{gauge fixing} Let us use the same gauge fixing as in the
Yang-Mills theory. For this we need to choose a metric on $M$.

Since classical solutions in the Chern-Simons theory are flat
connections, the covariant derivative $d_A=d+ A$ is a differential,
i.e. $d_A^2=0$ (twisted de Rham differential) acting on $\g$-valued
forms on $M$ \footnote{Because a principal $G$-bundle over any
compact oriented 3d-manifold is trivializable, we choose a
trivialization and identify $\Omega(M, ad(E))$ with $\Omega(M,\g)$.}.
Denote the cohomology spaces by $H^i_A$. Because of the Poincar{\'e}
duality we have natural isomorphisms $H^0_A\simeq H^3_A$ and
$H^1_A\simeq H^2_A$.

In a neighborhood of a classical solution $A$, connections
can be written as $A+\alpha$ where $\alpha$ is a $\g$-valued $1$-form on $M$.  As in the Yang-Mills theory, the {\it Lorenz gauge condition} for such connections is:
\[
d_A^*\alpha=0.
\]
We will use this gauge condition in the rest of the paper.

\subsection{The Faddeev-Popov action in the Chern-Simons theory}
According to our finite-dimensional example we should add fermionic
ghost fields $c(x)$ and $\bc(x)$ and the Lagrange multipliers
$\lambda(x)$ to the action, if we want to define Feynman diagrams in
this gauge theory. However, as we argued in Section
\ref{FP-gost-restr}, the gauge condition can be chosen linearly near
each critical point of the action, and therefore when we 
use Lorenz gauge condition we can use the
version without Lagrange multipliers. In this case we just have to
add fermionic ghost fields to the action.

According to the rules described in Section \ref{FP-gost-restr}, the Faddeev-Popov action for the Chern-Simons theory is the sum of the classical Chern-Simons action and the ghost
terms which are identical to those for the Yang-Mills theory:
\begin{multline}\label{FPCS}
CS_A(\alpha)=CS(A)+\int_M\frac{1}{2}tr\left(\alpha\wedge d_A\alpha-\frac{2}{3} \alpha\wedge \alpha\wedge \alpha\right)\\-
\frac{ih}{2} \int_M *d_A\bc\wedge d_A c-
\frac{ih}{2} \int_M *d_A\bc\wedge [\alpha,c],
\end{multline}
where $h$ stands for $\frac{1}{k}$, and $d_A^*\alpha=0$. We will focus in the discussion below mostly on the special case of
isolated flat connections, when $H^1_A=\{0\}$.
Quite remarkably \cite{AS1}, the field $\alpha$ and the ghost fields in the Chern-Simons theory can be  combined
into one odd ''super-field``:
\[
\Psi=c+\alpha+ ih *d_A\bar{c}.
\]
Here $c$, $\alpha$, and $*d_A\bar{c}$ are $0$, $1$, and $2$ forms
respectively. The action (\ref{FPCS}) can be written entirely in
terms of $\Psi$ \footnote{This form of the Faddeev-Popov action for the Chern-Simons theory has a simple explanation in the framework of the Batalin-Vilkovisky formalism, see for example \cite{CM}.
However, we will not discuss it in these notes.}:
\[
CS_A(\alpha)=CS(A)+\frac{1}{2}\int_M \tr \left(\Psi\wedge d_A\Psi -\frac{2}{3} \Psi\wedge\Psi\wedge\Psi\right).
\]
The quadratic part of the action is the de Rham differential
twisted by the flat connection $A$.

If $H_A^2(M, \g)=\{0\}$ (equivalently, $H^1_A=\{0\}$) the gauge condition
$d^*_A\alpha=0$ together with the special
form of the last term in $\Psi$,  is
equivalent to  $d^*_A\Psi=0$. The inverse is also true:
$d^*_A\Psi=0$ implies $d^*_A\alpha=0$ together with
$\Psi^{(2)}$ being the Hodge dual to an exact form.

The quadratic part of this action is
$(\Psi, *d_A\Psi)$ where
\[
(\Phi, \Psi)=\int_M \tr (\Phi\wedge *\Psi).
\]
The surface of the constraint $d_A\Psi=0$ is the super-space
$\Omega^0(M,\g)\oplus \Ker((d_A^*)^*_0)\oplus \Im((*d_A)_0)\subset
\Omega(M,\g)[1]$ where the first and the third summands are
odd and the second is even. The operator $D_A=*d_A+d_A*$ restricted
to this subspace describes the quadratic part of the action.
Indeed, we have
\[
\int_M \tr ( \Psi\wedge d_A\Psi)=\frac{1}{2}(\Psi,
(*d_A+d_A*)\Psi)\,.
\]
The operator $D_A$ maps even forms to even and odd form to odd,
$D_A: \Omega^i\to \Omega^{2-i}\oplus \Omega^{4-i}$. It plays a prominent place in index theory
\cite{AS}. It can be considered as a {\it Dirac operator}
in a sense that
\[
D_A^2=\Delta_A=d_A^*d_A+d_A d_A^*,
\]
where $\Delta_A$ is {\it Hodge Laplace operator}. The operator $D_A$ effectively appears in the quadratic part being restricted to odd forms. This operator will be denoted by
\[D_A^-: \begin{cases} \Omega^1\to \Omega^1\oplus \Omega^3 \\
    \Omega^3\to \Omega^1
         \end{cases}.
\]

Now the question is whether the operator $D_A^-$
is invertible on the surface of the constraint.
In other words, whether the Lorenz gauge is really a cross-section through gauge orbits.

\subsubsection{The propagator} \index{propagator} First, assume that the complex $(\Omega^i(M,\g), d_A)$
is {\it acyclic}, i.e. $H^i(M,\g)=\{0\}$ for all $i=0,1,2,3$ (by
Poincar{\'e} duality $H^i\simeq H^{3-i}$, so it is enough to assume
the vanishing of $H^0$ and $H^1$). In this case, the representation
of $\pi_1(M)$ in $G$ defined by holonomies of a flat connection $A$
is irreducible (implied by $H^0=\{0\}$) and isolated (implied by
$H^1=\{0\}$).

Since
the spaces $H^i$ can be naturally identified with harmonic
forms and therefore with zero eigenspaces of Laplace operators, in this case all Laplace operators are
invertible and so is $D_A$.
Denote by $G_A$ the inverse to $\Delta_A$, i.e. the Green's function, then
\[
P_A={(D_A^-})^{-1}=D_A^- G_A= G_A D_A^-\, .
\]
Thus, in this case the quadratic part is non-degenerate and we can
write contributions from Feynman diagrams as multiple integrals of
the kernel of the integral operator $({D_A^-})^{-1}$. The analysis
of the contributions of Feynman diagrams to the partition function
was studied in this case by Axelrod and Singer in
\cite{AS1}, \cite{AS2}, and by Kontsevich \cite{Kon}.

Another important special case arises when the flat connection is reducible but still isolated. For example, a trivial
connection for rational homology spheres \cite{BC1}, \cite{BC2}, \cite{C} has such property. In this case, we still have $H^1=H^2=\{0\}$
and the Lorenz gauge for $\alpha$ together with the exactness of $*\Psi^{(2)}$ is still equivalent to the Lorenz gauge for
$\Psi$, i.e. $d^*\Psi=0$.
However now there are harmonic forms in $\Omega^0(M)$
and $\Omega^3(M)$ corresponding to the fundamental class of $M$ and because of this, $D^-_A$ is not invertible on the space of all forms.

Nevertheless, in this case (and in a more general case when $H^1\neq \{0\}$) one can construct an operator which is ''almost inverse`` to $D_A^-$. Such an operator is determined by the
chain homotopy $K: \Omega^i\to \Omega^{i-1}$ and the Hodge decomposition of $\Omega$. For  details about such operator $P$ see \cite{AS1}, \cite{AS2}, \cite{BC1}, \cite{BC2}, \cite{C} and Section 3.2 of \cite{CM}.

An important example of a rational homology sphere is $S^3$ itself.
In this case, the inverse to $D^-$ for trivial a connection can be constructed explicitly by puncturing of $S^3$ at one point (the infinity). The punctured $S^3$ is homeomorphic to $\RR^3$ where the fundamental class is vanishing and
$D^-$ is invertible. The restriction of $(D^-)^{-1}$ to 1-forms is the integral operator with the kernel
\begin{equation}\label{flat-prop}
\omega(x,y)=\frac{1}{8\pi}\sum_{ijk=1}^3\epsilon^{ijk}\frac{(x-y)^i dx^i\wedge dy^k}{|x-y|^3}I,
\end{equation}
where $\epsilon^{ijk}$ is the totally antisymmetric tensor with $\epsilon^{123}=1$, and $I$ is the identity in $\operatorname{End}(V)$. It acts on the form $\sum_i \alpha_i(x)dx^i$ as
\begin{equation}\label{prop-flat}
\omega\circ\alpha(x)=\frac{1}{8\pi}\sum_{ijk=1}^3 \varepsilon^{ijk}\int_{\RR^3}\frac{(x-y)^i }{|x-y|^3}\alpha_j(x)d^3y dx^k\,.
\end{equation}

In all cases, the propagator $P_A$ is defined as the restriction of the ''inverse`` to $D^-_A$  (a chain homotopy, when $D^-$ is not invertible)
to one-forms. It is an integral operator with the kernel
being an element of the skew-symmetric part of $\Omega^2(M\times M, \g\times \g)$. If $e_a$ is an orthonormal basis in $\g$ and $x^i$
are local coordinates, we have
\[
P_A(x,y)=P^{ab}_A(x,y)_{ij} e_a\otimes e_b dx^i\wedge dy^j, \ \
 \ \  P^{ab}_A(x,y)_{ij}=P^{ba}_A(y,x)_{ji}\, .
\]

\subsection{Vacuum Feynman diagrams and invariants of $3$-manifolds}
\index{invariants of $3$-manifolds}

\subsubsection{Feynman diagrams} As in other examples of quantum field theories such as the
scalar Bose field and the Yang-Mills field we want to
define the semiclassical expansion of the partition function and
of correlation functions imitating the semiclassical expansion of
finite-dimensional integrals.

Following this strategy and the computations of the
Faddeev-Popov action for the Chern-Simons theory in Lorenz gauge
presented above, it is natural to define the partition function $Z(M)$ (the ''integral`` (\ref{CS-pathint})) for the Chern-Simons theory as the following combination of formal power series
\begin{equation}\label{CS-pert-series}
\sum_{[A]}\exp\left(i\frac{CS_M(A)}{h}+\frac{i\pi}{4}\eta(A)\right)
|{\det}'(D_A^-)|^{-1/2}{\det}'(\Delta_A^0) \left(1+\sum_{n\geq
1}(ih)^nI^{(n)}_A(M,g)\right),
\end{equation}
where $h$ stands for $k^{-1}$,  $\det'$ are regularized determinants
of corresponding differential operators, $\Delta_A^0$ is the
Laplace-Beltrami operator acting on $C^\infty(M,\mathfrak{g})$,
$D_A^-$ is the operator $D_A$ acting on odd forms, and $\eta(A)$ is
the index of the operator $D_A^-$. The sum is taken over gauge
classes of flat connections on $M$ (we assume that there is a finite
number of such isolated flat connections). The $n$-th order
contribution is given by the sum of Feynman diagrams
\begin{equation}\label{CS-ferm}
I^{(n)}_A(M,g)=\sum_{\Ga, -\chi(\Ga)=n} \frac{I_A(D(\Ga),M,g)(-1)^{c(D(\Ga))}}{|\operatorname{Aut}(\Ga)|}.
\end{equation}
In the Chern-Simons case, these are graphs with $2n$ vertices (each of them being $3$-valent). The contribution $I_A(D(\Ga),M,g)$ is an appropriate trace
of the integral over $M^m$ of the product of propagators. In other words this is the contribution from the Feynman diagram $D(\Ga)$
with weights from Fig. \ref{CS-F-diag}.\footnote{ The weights in Feynman diagrams for the Chern-Simons theory are the same
as we would have without the ghost fields (without the Faddeev-Popov
determinant). For the Chern-Simons theory the
ghost fields change the bosonic Feynman diagrams (which
we would have in the naive perturbation theory) to the
fermionic one (with the sign $(-1)^{c(D(\Ga))}$). It happens because $\Psi$ is an odd field and therefore the Feynman diagrams have fermionic nature.
The orientation of graphs used in \cite{Kon} is another way to
encode the fermionic nature of Feynman diagrams for the
Chern-Simons theory.

With the fermionic sign the sum of Feynman diagrams is finite in
each order \cite{AS1}. Without this sign the sum would diverge
because of the singularity of the propagator at the diagonal. It is
similar to the effect of ghost fields in the Yang-Mills theory.
Without ghost fields the Yang-Mills theory is not renormalizable.
With ghost fields, as it was shown by t'Hooft it becomes
renormalizable.}\label{foot_feyn}

\begin{figure}[htb]
\sidecaption[t]
\includegraphics[height=3cm,width=6cm]{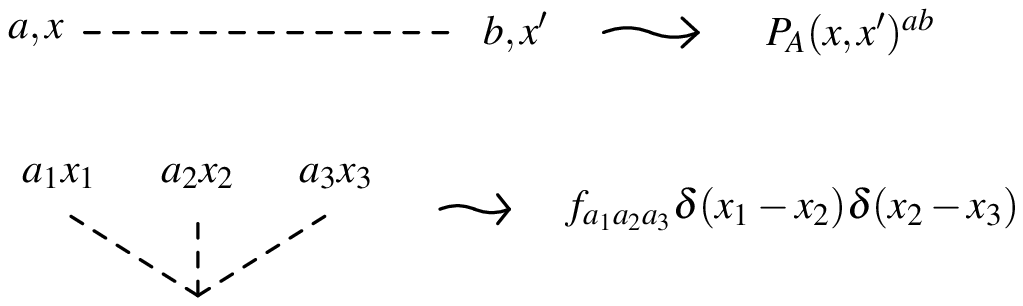}
\caption{Weights in Feynman diagrams for the Chern-Simons theory,
i.e propagators and vertices for the $\Psi$-field}
\label{CS-F-diag}
\end{figure}

Because in this case we have only $3$-valent vertices,
only two first order diagrams Fig. \ref{ord-1} will survive.
Among these two, only the ''theta graph`` will give a non-zero
contribution due to the skew-symmetry of the propagator.
The contribution from the theta graph is
\begin{equation}\label{CS-1-ord}
\int_M\int_M \sum_{\{a\}, \{b\}} f_{a_1a_2a_3}f_{b_1b_2b_3}
P^{a_1b_1}(x,y)P^{a_2b_2}(x,y)P^{a_3b_3}(x,y) dxdy.
\end{equation}

According to what we expect from the heuristic formula (\ref{CS-pathint}),
the expression  (\ref{CS-pert-series}) should depend only on
the homeomorphism class of $M$ and should not depend on the choice of the metric (gauge condition). But first of all, we should make
sure that every term in this series is defined. The problem
is that individual integrals in the definition of $I^{(n)}$
diverge.

As noticed in Footnote 12,
a remarkable property of
Feynman diagrams in the Chern-Simons case is that the sum of Feynman
diagrams of any given order is finite. It is relatively easy to see
that (\ref{CS-1-ord}) is finite because of the skew-symmetry of the
propagator and because it is asymptotically equivalent to
(\ref{prop-flat}) near the diagonal (i.e. when $x\to y$). The
finiteness of the sum of Feynman diagrams in each order was proven
in all orders by Axelrod and Singer in \cite{AS1}\cite{AS2} for
acyclic connections, and by Kontsevich in\cite{Kon} for trivial
connections in rational homology spheres. This illustrates that the
Chern-Simons theory is very different from the Yang-Mills theory
where the renormalization procedure is necessary.

\subsubsection{Metric independence} Now let us focus on the metric dependence
of (\ref{CS-pert-series}). Because we expect the quantum field
theory to be topological, the leading terms and each coefficient
in the expansion in the powers of $h$ should not depend on the metric.
First, assume that $A$ is an isolated irreducible flat connection.

The most singular term in the exponent is $CS_M(A)$ which
is clearly metric independent. Taking into account that $\Delta=
D^2$ and the natural isomorphism $\Omega^3\simeq \Omega^0$, the absolute value of the determinant of $D^-_A$ can be written
as
\[
|{\det}'(D_A^-)|={\det'}(\Delta^1_A)^{\frac{1}{2}}{\det'}(\Delta^3_A)^{\frac{1}{2}}=
{\det'}(\Delta^1_A)^{\frac{1}{2}}{\det'}(\Delta^0_A)^{\frac{1}{2}}\,,
\]
where $\Delta^i_A$ is the action of the Laplacian twisted by $A$ on $i$-forms. Using this identity we can rearrange the determinants as
\[
|{\det'}(D_A)|^{-1/2}\det'(\Delta_A^0)=\frac{{\det'}(\Delta_A^0)^{\frac{3}{4}}}
{{\det'}(\Delta_A^1)^{\frac{1}{4}}}\,.
\]
This expression is exactly the square root of the {\it Ray-Singer
analytical torsion} \cite{RS}, which is also the {\it Reidemeister
torsion}, and is known to be independent of the metric \footnote{When
$A$ is an isolated irreducible flat connection, the Ray-Singer
torsion is defined as the positive number $\tau(M,A)$ such that
$$
\tau(M,A)=\prod_{i\geq 1} {\det'}(\Delta^i_A)^{i(-1)^{i+1}/2}\, .
$$
Here and in the main text ${\det'}(\Delta)$ is the zeta function regularization of the determinant: ${\det'}(\Delta)=\exp(-\zeta_\Delta'(0))$,  where
\[
\zeta(s)=\Tr(\Delta^{-s})=\frac{1}{\Gamma(s)} \int_0^\infty t^{s-1} \Tr(e^{t\Delta})dt,
\]
Taking into account that for Riemannian manifolds
we have natural isomorphisms $\Omega^0\simeq \Omega^3$ and
$\Omega^1\simeq \Omega^2$ we obtain
$$
\tau(M,A)={\det}'(\Delta_A^0)^{\frac{3}{2}}
{\det}'(\Delta_A^1)^{-\frac{1}{2}}.
$$
When the $H^i$ are not all zeroes, $\tau(M,A)^{\frac{1}{2}}$ can be
regarded as a volume element on ''zero`` modes, i.e. on the space
$H^0\oplus H^1$.}.

The exponent $\frac{i\pi}{4}\eta(A)$, which involves the index of $D_A^-$
can be written as
\begin{multline}
\frac{2\pi\eta(A)}{8}=d\frac{2\pi\eta(g,M)}{4}+c_2(G)CS(A)-\frac{2\pi}{4}I_A-
\frac{d\pi(1+b^1(M))}{4}\\
+\frac{2\pi(\dim(H^0)+\dim(H^1))}{8} \ \ (mod \  2),
\end{multline}
Here $\eta(g,M)$ is the {\it index} of the operator $D=*d+d*$ acting on odd forms on $M$, $d=\dim(G)$,  $c_2(G)$ is the value of the
{\it Casimir element} for $\mathfrak{g}=Lie(G)$ on the adjoint
representation (also known as the dual {\it Coxeter number} $h^\vee$ for the
appropriate normalization of the {\it Killing form} on $\mathfrak{g}$),
and $b^1(M)$ is the first {\it Betti number} for $M$. The quantity $I_A\in \ZZ/8\ZZ$ is the {\it spectral flow} of the operator
\[
\left(\begin{array}{cc} *d_{A_t}& -d_{A_t}*\\
d_{A_t}* & 0
\end{array}\right)
\]
acting on $\Omega^1(M,\g)\oplus \Omega^3(M,\g)$. Here $A_t$, $t\in [0,1]$, is a path in the space of connections joining $A$ with the
trivial connection. The spectral flow $I_A$ depends neither on the metric on $M$ nor on the choice of the path.

The index $\eta(g,M)$ depends on the metric $g$ on $M$. Recall that
a {\it framing} of $M$ is the homotopy class of a trivialization of
the tangent bundle $TM$. Given a  framing $f: M\to TM$ of $M$ we can
define {\it the gravitational Chern-Simons action}
\index{gravitational Chern-Simons action}
\begin{equation}\label{GCS}
I_M(g,f)=\frac{1}{4\pi}\int_M f^*Tr(\omega\wedge d\omega -\frac{2}{3}\omega\wedge\omega\wedge\omega),
\end{equation}
where $g$ is the metric on $M$, $\omega$ is the Levi-Civita
connection on $M$, and the integrand is the pullback via $f^*$ of
the Chern-Simons form on $TM$.

According to the Atiyah-Patodi-Singer theorem the
expression
\[
\frac{1}{4}\eta(g,M)+\frac{1}{12}\frac{I_M(g,f)}{2\pi}
\]
depends only on the homeomorphism class of the manifold
$M$ with the framing $f$, but not on the metric,
and this is true for any framing $f$.

These arguments suggest \cite{W}, \cite{FG}  that for manifolds with only irreducible and isolated flat connections, the leading term in the expression (\ref{CS-pert-series}) should be proportional to
\begin{multline}
\exp\left( d\frac{i\pi}{4}\eta(g,M)+i\frac{d}{24}I_M(g,f)-\frac{di\pi}{4}\right)\\
\sum_{[A]}
\exp\left(-\frac{2\pi i I_A}{4}+i(\frac{1}{h}+c_2(G))CS_M(A)\right)
\tau(M,A)^{1/2}(1+O(1/k)),
\end{multline}
where $\tau(M,A)$ is the Ray-Singer torsion. This expression differs from the original guess (\ref{CS-pert-series}) by the extra
factor $\exp(i\frac{d}{24} I_M(g,f))$.

Let us emphasize that this formula is {\it not a computation of the
path integral}, as there is nothing to compute. It is a rearrangement and adjustment of the natural guess for the leading
terms of the semi-classical expansion of the quantity to be defined.
The adjustment was made on the base of the concept that the
expression should not depend on the metric. Remarkably, at the
end it does not depend on the metric, though it still depends on the
framing.

Now let us look into higher order terms.

In the finite-dimensional case, when Feynman diagrams represent an
asymptotic expansion of an existing (convergent) integral, the sum of
Feynman diagrams in each order does not depend on the choice of the gauge
condition simply by the nature of these coefficients.

In the {\it infinite-dimensional case} we are defining the integral
as a sum of Feynman diagrams. Therefore, the independence of the
sum of Feynman diagrams on the choice of the gauge condition (a
metric on $M$ in the case of the Lorenz gauge condition for the Chern-Simons
theory) should be checked independently in each order. This was done
by Axelrod and Singer in \cite{AS1}, \cite{AS2} for acyclic
connections and by Bott and Cattaneo \cite{BC1}, \cite{BC2} for
trivial connections and rational homology spheres. One of the
important tools for the proof of such fact is the {\it graph complex} by
Konstevich \cite{Kon}.

More precisely the following has been proven. First write
the sum of higher order contributions as
\[
1+\sum_{n\geq 1}(ih)^nI^{(n)}_A(M,g)=\exp\left(\sum_{n\geq 1}(ih)^nJ^{(n)}_A(M,g)\right),
\]
where
\[
J^{(n)}_A(M,g)=\sum^{(c)}_{-\chi(\Ga)=n}
\frac{I_A(D(\Ga),M,g)(-1)^{c(D(\Ga))}}{|\Aut(\Ga)|}\,.
\]
Here the sum is taken over connected graphs only.
As it follows from \cite{AS1}, \cite{AS2}, \cite{BC1}, \cite{BC2}
this expression can be written as
\[
J^{(n)}_A(M,g)=F^{(n)}_A(M,f)+ \beta_n I_M(g,f),
\]
for some $F^{(n)}_A(M,f)$ and constants $\beta_n$. Here $I(g,f)$ is the gravitational Chern-Simons action (\ref{GCS}).

Thus, the sum of contributions of connected Feynman diagrams of
fixed order, after the substraction of the
gravitational Chern-Simons action with an appropriate numerical coefficient, does not depend on the metric, and, therefore, is an invariant of framed rational homology spheres
(in the works of Bott and Cattaneo) or of
a 3-manifold with an acyclic flat connection in a trivial principal $G$-bundle over it (in the works of Axelrod and Singer, and Kontsevich).

Finally, all these results can be summarized as the following
proposal for the partition function of the semiclassical
Chern-Simons theory. It depends on the framing and is proportional
to

\begin{multline}\label{CS-closed}
\exp\left(
c(h)(\frac{i\pi}{4}\eta(g,M)+i\frac{1}{24}I_M(g,f))\right)e^{-\frac{id\pi(1+b^1(M))}{4}}\sum_{[A]}
e^{i(\frac{1}{h}+c_2(G))CS_M(A)}\\ \exp({-\frac{2\pi i
I_A}{4}})\tau(M,A)^{1/2} \exp\bigl(\sum_{n\geq
1}(ih)^nF^{(n)}_A(M,f)\bigr),
\end{multline}
where $c(h)=d+O(h)$. Witten suggested \cite{W} the exact form of $c(h)$:
\[
c(h)=\frac{d}{1+hc_2(G)}=\frac{kd}{k+h^\vee},
\]
where $k=\frac{1}{h}$. This is the central charge of the corresponding {\it Wess-Zumino-Witten theory}.

In order to define the full TQFT from this proposition one should
define the partition function in the case when flat connections are not necessarily irreducible and when they are not isolated. For most the recent progress in this direction see \cite{CM}. Also, in order to have a TQFT we should define
partition functions for manifolds with boundaries. In the semiclassical framework this is an open problem.

\subsection{Wilson loops and invariants of knots}
Arguing ''phenomenologically`` one should anticipate that expectation values of topological\footnote{Topological observables do not require a metric in their definition.} gauge invariant observables in Chern-Simons theory, which do not require metric in their definition, should depend only on topological data and, therefore, give some {\it topological
invariants}.

\subsubsection{Wilson graphs}
\index{Wilson graphs} An example of topological  observables are
Wilson loops (\ref{W-loop}), or, more generally, Wilson graphs. Let
us clarify the notion of the Wilson loop observable in the
perturbative Chern-Simons theory. Wilson loops are defined for a
collection of circles embedded into $M$ otherwise known as a {\it
link}. Our goal is to define the power series which would be similar
to the perturbative expansion (\ref{as-FP}), as (\ref{CS-closed}) is
similar to the perturbative expansion (\ref{red-FP}). Most
importantly, such power series should not depend on the choice of a
metric on $M$ (the choice of the gauge condition). As we have seen
above, this is possible but one should  choose a framing $f: M\to
TM$ of the 3-manifold.

A framed Wilson graph observable (or simply a Wilson graph)
is a gauge invariant functional on connections
defined as follows. Let $\Ga$ be a framed graph\footnote{
The embedding $C\subset M$ induces the embedding $TC \subset TM$. Therefore a framing on $M$ induces a framing on $C$,
i.e the mapping $C\to (T_CM/TC)^{perp}$. A metric on $M$
defines the splitting $T_CM=NC\oplus TC$ where $NC$ is a
normal bundle to $C$. A framing $f: M\to TM$ defines
the framing $f_C: C\to NC$ of $C$ by attaching a normal vector $f_C(x)\in N_xC$ for every $x\in C$.}
embedded in $M$. Here by the framing we mean a section
of the co-normal bundle $x\in \Ga\to TM/T_x\Ga$ for
a generic point $x\in \Ga$ which agrees on vertices.

Framing, together with the orientation of $M$  defines a cyclic ordering of edges adjacent to each vertex. It is illustrated
in Fig. \ref{vert-frame}.

\begin{figure}[htb]
\sidecaption
\includegraphics[height=3cm,width=6cm]{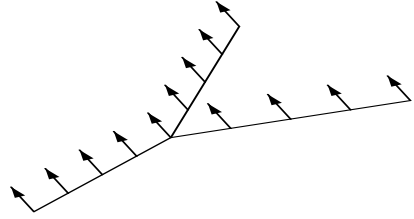}
\caption{Parallel framing at a trivalent vertex}
\label{vert-frame}
\end{figure}

To define a Wilson graph we should make the following choices:

\begin{enumerate}
\item  Choose a total ordering of edges adjacent to each
vertex which agrees with the cyclic ordering defined by the framing.
\item Choose an orientation of each edge.
\item Choose a total ordering of vertices of $\Ga$.
\item Choose a finite-dimensional representation $V$ for
each edge of $\Ga$.
\item Choose a $G$-invariant linear mapping $\nu: \CC\to V_1^{\ep_1}\otimes \dots \otimes V_k^{\ep_k}$ for each vertex. Here numbers $1,\dots, k$ enumerate edges adjacent
to the vertex, $\ep_i=+ $ if the edge $i$ is incoming to the vertex, $\ep_i=-$ if the edge $i$ is outgoing from the vertex,
$V_i^+=V_i$, $V_i^-=V_i^*$,  $V_i$ is the representation assigned to the edge $i$, and $V_i^*$ is its dual.
\end{enumerate}

As in the case of Feynman diagrams, the ordering of vertices, and on the edges adjacent to each vertex, defines a perfect matching on endpoints of edges. Choose such total ordering.

Use the coloring of edges by finite-dimensional $G$-modules and
the orientation of edges to define the space $ V_{a_1}^{\alpha_1}
\otimes V_{a_2}^{\alpha_2}
\otimes V_{a_3}^{\alpha_3}
\otimes V_{b_1}^{\beta_1}
\otimes V_{b_1}^{\beta_2}
\otimes \cdots $. Here indices $1,2,\cdots$ enumerate vertices,
letters $a_i,b_i,c_i,\dots$ enumerate edges adjacent to the vertex $i$, and $\alpha, \beta, \cdots =\pm$ indicate the orientations of edges $a,b,c,\cdots$ ($+$ if the edge is incoming, and $-$ if the edge is outgoing). The number of factors in the tensor product
is equal to the number of endpoints of edges.

The coloring of vertices defines the vector
\[
\nu_1\otimes \nu_2 \otimes \cdots \in V_{a_1}^{\alpha_1}
\otimes V_{a_2}^{\alpha_2}
\otimes V_{a_3}^{\alpha_3}
\otimes V_{b_1}^{\beta_1}
\otimes V_{b_1}^{\beta_2}
\otimes \cdots.
\]

The holonomy $h_e(A)$  along the edge $e$ is an element of $\End(V_e)$,
and therefore, it is a vector in $V_e\otimes V_e^*$, where
$V_e$ is the finite-dimensional $G$-module assigned to the edge. The tensor product of holonomies defines a vector $\otimes_e h_e(A)$ in the space
dual to $V_{a_1}^{\alpha_1}
\otimes V_{a_2}^{\alpha_2}
\otimes V_{a_3}^{\alpha_3}
\otimes V_{b_1}^{\beta_1}
\otimes V_{b_1}^{\beta_2}
\otimes \cdots.$.

The Wilson graph observable is the functional
on the space of connections defined
as
\[
W_\Ga(A)=\left< \otimes_e h_e(A),\nu_1\otimes \nu_2 \otimes \dots\right>
\]

Here is an example of the Wilson graph observable for the ''theta
graph``:
\[
\sum_{i_1,i_2,i_3}(h_{e_1}(A))^{i_1}_{j_1}(h_{e_1}(A))^{i_1}_{j_1}
(h_{e_1}(A))^{i_1}_{j_1}\nu_{i_1,i_2,i_3}\nu_{j_1,j_2,j_3}.
\]
The indices $i_k,j_k$ enumerate a basis in the representation
$V_k$ assigned to the edge $k=1,2,3$ and $\nu, \mu$ are $G$-invariant vectors in the corresponding tensor products.
Here we used an orthonormal basis in $\g$ which explains upper and lower indices.

\subsubsection{Feynman diagrams for Wilson graphs} As in the case of the partition function, define
the expectation value (\ref{CS-Wloop}) of the Wilson graph $\Ga$ as
a combination of formal power series, similar to the formula
(\ref{as-FP}) for the asymptotic expansion of corresponding finite-dimensional integrals.

Taking into account all we know for the partition
functions of the Chern-Simons theory we arrive to
the following proposal. The semiclassical ansatz for
the expectation value of the Wilson graph $W_\Ga$ is
\begin{multline}\label{CS-pert-loops}
\sum_{[A]}\exp\left(i\frac{CS(A)}{h}+\frac{id\pi}{4}\eta(A)\right) |{\det}'(D_A)|^{-1/2}\\
{\det}'((\Delta_A)_0)
\left(W_\Ga(A)+\sum_{n\geq 1}(ih)^nI^{(n)}_A(M, \Ga)\right).
\end{multline}
Here we assume that all flat connections are irreducible and isolated. All quantities are the same as in (\ref{CS-pert-series})
except
\[
I^{(n)}_A(M, \Ga)=\sum_{\Ga'\,,\, \chi(\Ga)-\chi(\Ga')=n}
\frac{I_A(\Ga',\Ga)}{|\Aut(\Ga)|}.
\]
The Feynman diagram rules in the presence of Wilson graphs are
essentially the same as for the partition function with weights
given in Fig. \ref{CS-F-diag}. The difference is that now there are
two types of edges, and two types of propagators (linear operators
assigned to edges). As for the partition function we have dashed
edges with 3-valent vertices. But now we also have solid edges, see
an example in Fig. \ref{W-examples}, vertices where only solid edges
meet, and vertices where two solid edges (with opposite orientations)
meet a dashed edge. The subgraph formed by solid edges is always
$\Ga$. The weights of vertices where only solid edges meet is given
by the coloring of this edge in $\Ga$. The weights of vertices where
two solid edges meet a dashed edge  are
described in Fig. \ref{w-weight}. The weight of a solid edge is
the parallel transport for the connection $A$ along this
edge.

\begin{figure}[htb]
\sidecaption
\includegraphics[height=3cm,width=6cm]{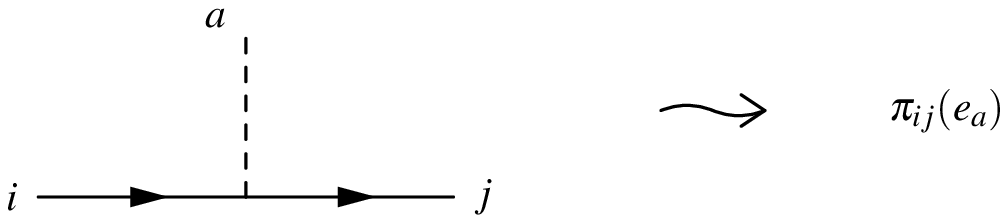}
\caption{Weights of trivalent vertices where two solid
edges meet one dashed edge}
\label{w-weight}
\end{figure}

\begin{figure}[htb]
\sidecaption[t]
\includegraphics[height=3cm,width=6cm]{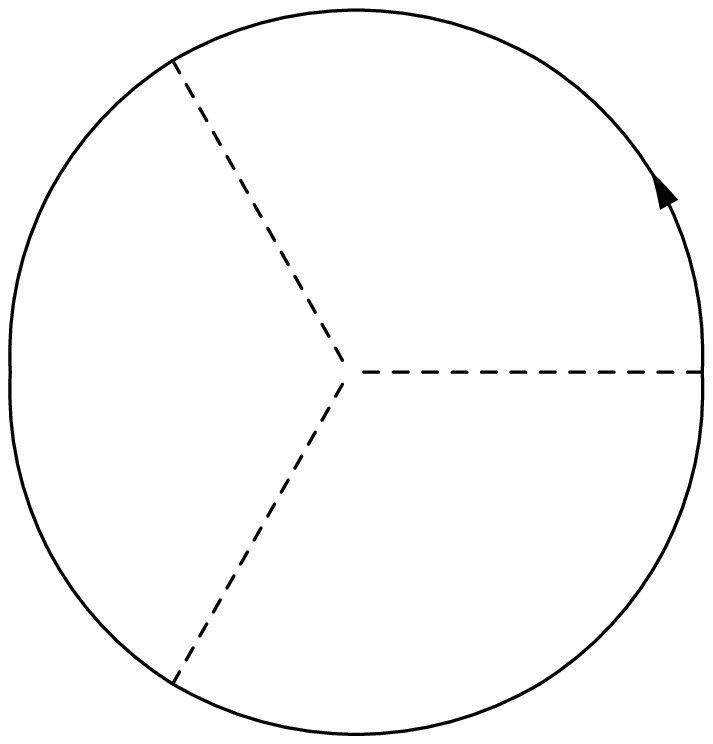}
\caption{An example of an order one graph}
\label{W-examples}
\end{figure}

One can show \cite{DBN}, \cite{GMM}, \cite{AF}, \cite{BT}, \cite{DT},
\cite{BC1}, \cite{BC2} that the sum of integrals corresponding to Feynman diagrams of order $n$ is finite for each $n$. Similarly to the vacuum partition function from the previous section, the semiclassical ansatz for the expectation
value of the Wilson graph depends on the framing, but remarkably not on the metric. When flat connections are irreducible
and isolated we arrive at the following expression

\begin{multline}\label{CS-closed}
\exp\left(c(h)(\frac{i\pi}{4}\eta(g,M)+\frac{i}{24}I_M(g,f))-\frac{id\pi(1+b^1(M))}{4}\right)\\
\sum_{[A]}
\exp\left(i(\frac{1}{h}+c_2(G))CS_M(A)-\frac{2\pi i I_A}{4}\right)\tau(M,A)^{1/2}\\
\left(W_\Ga(A)+
\sum_{n\geq 1}(ih)^nJ^{(n)}_A(M,\Ga,f)\right).
\end{multline}
Here the coefficients $J^{(n)}_A(M, \Ga, f)$ do not
depend on the metric but depend on the framing $f$ of $M$.
This formula defines the path integral semiclassically. Let us
emphasize again, that it is
not a result of computation of an integral. It is a definition, modeled after the semiclassical expansion of integrals
in terms of Feynman graphs. A remarkable {\it mathematical} fact is that every term is defined (the integrals do not diverge), and that it does not depend on the metric.

More careful analysis includes powers of $h$. A conjecture for counting powers of $h$ when $H_A^0, H_A^1 \neq \{0\}$ was proposed in \cite{FG}, \cite{J}, \cite{Ro1}. It agrees with the finite-dimensional analysis from previous sections and states that, in general, we should expect that the partition function is proportional to

\begin{multline}
\exp\left(d\frac{i\pi}{4}\eta(g,M)+i\frac{c(h)}{24}I_M(g,f)-\frac{d\pi i(1+b^1(M))}{4}\right)\\
\sum_{A} (2\pi (k+h^\vee)^{\frac{\dim(H^0_A)-\dim(H^1_A)}{2}} \frac{1}{Vol(G_A)}
\\ \exp\left(i(k+h^\vee)CS_M(A)-\frac{2\pi i I_A}{4}-i\pi\frac{\dim(H^0_A)+\dim(H^1_A)}{2}\right)\\
\int_{M_A}\tau^{1/2}
\left(W_\Gamma(A)+\sum_{n\geq 1}(ih)^{n}J^{(n)}_A(M,\Ga,f)\right).
\end{multline}

Here the sum is taken over representatives $A$ of connected components $M_A$ of the moduli space of flat connections in a principal $G$-bundle over $M$. The torsion $\tau$ is an element of $\otimes_i det(H^i_A)^{\otimes (-1)^i}\simeq (det(H^0_A)\otimes det(H^1_A)^*)^{\otimes 2}$.
The Lie algebra $\g$ has an invariant scalar product and therefore $H^0_A\subset \g$ has an induced volume form. Pairing this
volume form with the square root of the torsion gives a volume
form on $H^1_A$. Assuming the connected component is smooth
we can integrate functions with respect to this volume form.
The factor $Vol(G_A)$ is the volume of the stabilizer
of the flat connection.

\subsection{Comparison with combinatorial invariants}

Invariants of $3$-manifolds with framed graphs also can be constructed combinatorially (as a {\it combinatorial topological quantum field theory}). In \cite{ReshTur} such invariants were constructed using modular categories and the representation of $3$-manifolds as a surgery on $S^3$ or on a handlebody along a framed link.
Another combinatorial construction, based on the triangulation, was developed in \cite{TV}. This construction uses a certain class of
monoidal categories which are not necessarily braided.

These two constructions are related:
\[
Z^{RT}_M({\mathcal C}) Z^{RT}_{\overline{M}}({\mathcal C})=Z^{TV}_M({\mathcal C})=Z^{RT}_M(D({\mathcal C})).
\]
Here $Z^{RT}_M({\mathcal C})$ is the invariant obtained by the surgery \cite{ReshTur}, $Z^{TV}_M({\mathcal C})$ is the invariant obtained by the
triangulation, and the category $D({\mathcal C})$ is the
center (the double) of the category ${\mathcal C}$,
see for example \cite{Ka}, and $\overline{M}$ is
the manifold $M$ with the reversed orientation.

Most interesting known examples of modular categories are quotient
categories of finite-dimensional modules over quantized universal
enveloping algebras at roots of unity, see
\cite{ReshTur}, \cite{HA}, \cite{GK}. Such categories are parametrized
by pairs $(\epsilon, \g)$, where $\epsilon=\exp(\frac{2\pi i m}{r})$
with mutually prime $m$ and $r$ and $\g$ is a simple Lie algebra.
Denote the truncated category of modules over $U_\epsilon(\g)$ by
${\mathcal C}_\ep(\g)$ (see \cite{ReshTur}, \cite{HA}, \cite{GK} for details). When $m=1$ and $r=k+c_2(\g)$ this category
is naturally equivalent to the braiding-fusion category of the WZW
conformal field theory at level $k$, i.e. to the category of
integrable modules over the affine Lie algebra $\hat{\g}$ at level
$k$ with the fusion tensor product \cite{KL}. This conformal field
theory is directly related to the Chern-Simons theory at level $k$.
The arguments in favor of this are not perturbative \cite{ADPW}.
They are based on ideas of geometric quantization.

For other values of $m$, the category ${\mathcal C}_\ep(\g)$ is also
equivalent to the braiding-fusion category of a conformal field
theory, but this conformal field theory is not directly related to
the Chern-Simons theory.

The main conjecture relating the combinatorial
and geometric approaches is that the {\it following power series
are identical}:

\begin{itemize}

\item The asymptotic expansion of the combinatorial
TQFT based on the category ${\mathcal C}_\ep(\g)$ when
$\ep=\exp(\frac{2\pi i}{k+c_2(\g)})$ and $k\to \infty$.

\item The semiclassical expansion for the Chern-Simons path integral in terms of Feynman diagrams.
\end{itemize}
Of course this is an outline of a number of conjectures rather than a conjecture. The main reason is that the semi-classical partition functions for the Chern-Simons theory in terms of Feynman integrals are not worked out  yet.

The precise statement about the correspondence between these formal power series was outlined in \cite{W} \cite{FG}, and then many results were obtained in \cite{J}, \cite{Ga}, \cite{Ro1}, \cite{Ro2}, \cite{AH}.

To compare these invariants one should first choose a canonical 2-framing on $M$ \cite{At}. The 2-framing on $M$ is a section of
$TM\times TM$. The Levi-Civita connection on $TM$ defined by the Riemannian structure on $M$ induces a connection on $TM\times TM$.
The canonical 2-framing defines the branch of the gravitational Chern-Simons action with the property
\[
d\frac{\pi}{4}\eta(g,M)+\frac{c(h)}{24}I_M(g,f)=0.
\]
One should expect that the choice of such 2-framing presumably fixes the framing
in higher order corrections, though this part is still conjectural.

When the moduli space of flat connections on a principal $G$-bundle over $M$ is a collection of isolated points, each corresponding to an irreducible flat connection, one
should expect the following:
\[
Z^{RT}_M\sim \frac{1}{|Z(G)|} \exp({-\frac{d\pi i}{4}})\sum_{[A]}
\exp({\frac{-2i\pi I_A}{4}}) \exp({(k+h^\vee)CS_M(A)})(1+O(1/k)),
\]
where $|Z(G)|$ is the number of elements in the center of
$G$, and $Z^{RT}$ is the combinatorial invariant corresponding
to the category ${\mathcal C}_\ep(\g)$. When connected components of the moduli space have non-zero-dimension and are smooth, the expected asymptotic behavior is
\begin{multline}
Z^{RT}_M\sim \exp\left(-\frac{d\pi i(1+b^1(M))}{4}\right)\sum_{[A]} (2\pi(k+h^\vee))^{\frac{\dim(H^0_A)-\dim(H^1_A)}{2}}\frac{1}{Vol(G_A)}
\\ \exp\left(i(k+h^\vee)CS_M(A)-\frac{2\pi i I_A}{4}-i\pi\frac{\dim(H^0_A)+\dim(H^1_A)}{2}\right)\\
\int_{M_A}\tau^{1/2}
W_\Gamma(A)\left(1+O(1/k)\right).
\end{multline}

Many examples confirming this prediction were analyzed in
\cite{FG}, \cite{J}, \cite{Ga}, \cite{Ro1}, \cite{AH}.

\printindex
\end{document}